\documentclass[12pt,a4paper]{article}
\usepackage{fullpage}
\usepackage[utf8]{inputenc}

\RequirePackage[round]{natbib}
\RequirePackage[colorlinks,citecolor=blue,urlcolor=blue, linkcolor=blue]{hyperref}
\RequirePackage{graphicx}
\usepackage{caption}
\usepackage{subcaption}
\usepackage{rotating} 
\usepackage{tikz} 
\usepackage{tikz-3dplot}
\usetikzlibrary{arrows, decorations.pathreplacing, positioning, shapes, arrows.meta}
\usepackage{authblk}

\usepackage{amsmath, amsthm, amssymb, mathtools}
\usepackage{bm, bbm, mathrsfs, bigints}
\usepackage{blkarray}

\usepackage[capitalise]{cleveref}

\usepackage{setspace}

\usepackage{booktabs, multirow}
\usepackage{multicol}
\usepackage{enumerate}
\usepackage{algorithm,algpseudocode}
\usepackage{comment}
\usepackage{lineno}
\usepackage[enable]{easy-todo}

\usepackage{titlesec}
\titleformat{\section}{\large\bfseries}{\thesection}{1em}{}
\titleformat{\subsection}{\normalsize\bfseries}{\thesubsection}{1em}{}
\usepackage{chngcntr}

\usepackage{mathtools}

\definecolor{blue-violet}{rgb}{0.54, 0.17, 0.89}
\definecolor{antiquefuchsia}{rgb}{0.57, 0.36, 0.51}
\definecolor{amethyst}{rgb}{0.6, 0.4, 0.8}
\definecolor{ao}{rgb}{0.0, 0.5, 0.0}
\definecolor{blue(ncs)}{rgb}{0.0, 0.53, 0.74}
\definecolor{dgreen}{rgb}{0.12, 0.3, 0.17}
\definecolor{cadmiumgreen}{rgb}{0.0, 0.42, 0.24}
\definecolor{darkolivegreen}{rgb}{0.33, 0.42, 0.18}
\definecolor{dartmouthgreen}{rgb}{0.05, 0.5, 0.06}

\newcommand{\tild}{\raise.17ex\hbox{ $\scriptstyle\sim$ }}

\DeclareMathOperator{\var}{\mathrm{var}}

\DeclareMathOperator{\E}{\mathbb{E}}
\DeclareMathOperator{\pr}{\mathbb{P}}
\DeclareMathOperator{\G}{\mathsf{G}}

\DeclareMathOperator*{\argmin}{arg\,min}
\DeclareMathOperator*{\argmax}{arg\,max}

\newcommand{\R}{\mathbb{R}}

\newcommand{\expo}{\mathrm{Exp}}

\newcommand{\N}{\mathcal{N}}

\newcommand{\pois}{\mathrm{Poisson}}

\theoremstyle{plain}
\newtheorem{theorem}{Theorem}
\newtheorem{corollary}{Corollary}
\newtheorem{proposition}{Proposition}
\newtheorem{lemma}{Lemma}

\newtheorem{assumption}{Assumption}

\theoremstyle{definition}

\theoremstyle{remark}

\Crefname{theorem}{Theorem}{Theorems}
\Crefname{corollary}{Corollary}{Corollaries}
\Crefname{proposition}{Proposition}{Propositions}
\Crefname{lemma}{Lemma}{Lemmas}
\Crefname{assumption}{Assumption}{Assumptions}
\Crefname{condition}{Condition}{Conditions}
\Crefname{definition}{Definition}{Definitions}
\Crefname{example}{Example}{Examples}
\Crefname{remark}{Remark}{Remarks}

\makeatletter
\newenvironment{assumpenum}
  {\renewcommand{\p@enumi}{\theassumption}\begin{enumerate}}
  {\end{enumerate}}
\makeatother

\providecommand{\keywords}[1]
{{\small \textbf{Keywords---}#1}}

\newcommand{\appendixnumbering}{\renewcommand{\thetheorem}{S\arabic{theorem}}\renewcommand{\thealgorithm}{S\arabic{algorithm}}\renewcommand{\thelemma}{S\arabic{lemma}}\renewcommand{\theexample}{S\arabic{example}}\renewcommand{\theassumption}{S\arabic{assumption}}\renewcommand{\theproposition}{S\arabic{proposition}}\renewcommand{\thecorollary}{S\arabic{corollary}}\renewcommand{\thedefinition}{S\arabic{definition}}\renewcommand{\thefigure}{S\arabic{figure}}\renewcommand{\thetable}{S\arabic{table}}}

\makeatletter
\renewenvironment{proof}[1][\relax]{\par
  \pushQED{\qed}\normalfont \topsep6\p@\@plus6\p@\relax
  \trivlist
  \item[\hskip\labelsep\itshape
    \ifx#1\relax \proofname\else\proofname{} #1\fi\@addpunct{.}]\ignorespaces
}{\popQED\endtrivlist\@endpefalse
}
\makeatother

\let \hat \widehat
\let \tilde \widetilde

\newcommand\given{\,|\,}
\newcommand{\Rom}[1]{\text{\MakeUppercase{\romannumeral #1}}}

\renewcommand{\G}{\mathcal{G}}
\newcommand{\F}{\mathcal{F}}

\renewcommand{\var}{\mathrm{Var}}

\newcommand{\ind}{\mathbbm{1}}

\newcommand{\sgn}{\mathrm{sgn}}

\newcommand\restr[2]{\left.\kern-\nulldelimiterspace #1 \vphantom{\big|} \right|_{#2}}
\newcommand{\DB}{\mathrm{DB}}

\begin{document}
\title{The Debiased Score Test:\\Hunt-and-test for Semiparametric Hypotheses}
\date{July 2026}

\author[1]{Aditya Dhawan \thanks{\texttt{ad950@cam.ac.uk}}}
\author[2]{F.~Richard Guo \thanks{\texttt{ricguo@umich.edu}}}
\author[1]{Rajen D.~Shah \thanks{\texttt{r.shah@statslab.cam.ac.uk}}}

\affil[1]{Statistical Laboratory, University of Cambridge, Cambridge, UK}
\affil[2]{Department of Statistics, University of Michigan, Ann Arbor, USA}

\maketitle

\begin{abstract}
The parametric score test assesses a hypothesis through derivatives of the log-likelihood, whose expectation vanishes under the null. When the parameter of interest is a regression function identified as a risk minimiser, we extend this idea to test whether it belongs to a given linear function class. This yields goodness-of-fit tests for common semiparametric regression models, including generalised additive and partially linear models. Suitably formulated, the framework also detects effect modifiers in observational studies. We propose a hunt-and-test strategy that splits the data into two: on one part, after fitting the null model, machine learning is used to identify a promising direction in the empirical scores; on the other, we test whether the score vanishes in that direction. To account for error in estimating the null model, we apply a debiasing correction based on a weighted least squares projection. We establish Type I error control under relatively mild conditions and show the test has power whenever the hunted direction is correlated with the true score. Simulations and real-data examples demonstrate favourable performance, including identifying effect modifiers in an HIV clinical trial and assessing an additive model for insurance claims. The methodology is implemented in the R package \texttt{dScoreTest}.
\end{abstract}


\section{Introduction} \label{sec:intro}
Models are the cornerstone of statistical inference.
They provide tractable representations of complex data-generating mechanisms, enabling interpretation, prediction, and uncertainty quantification.
Nevertheless, models are inevitably approximations.
Their usefulness depends on the discrepancy between the model and the true data-generating process being small relative to the precision of the conclusions drawn from them.
Otherwise, one risks obtaining highly precise inferences about quantities that lack meaningful interpretation under the true data-generating mechanism.
For example, fitting an overly simplistic model to a large dataset may produce very narrow confidence intervals around parameters that are scientifically uninformative because the model is seriously misspecified.

Assessing whether the data provide substantial evidence against a model is therefore a fundamental part of statistical practice.
The absence of such evidence is arguably a minimal requirement for a model to be useful.
Goodness-of-fit tests provide a formal and convenient framework for this assessment, helping to determine when a simple model must be replaced by a more flexible alternative in order to maintain an adequate representation of the data.

As discussed, larger datasets tend to demand more flexible models, but this brings some challenges for goodness-of-fit testing.
Considering Type I error control, it may not be straightforward to distinguish genuine incompatibility with the data-generating process from an apparent lack of fit due to the difficulties in estimating a flexible model.
On the other hand, one would like to have power against diverse and potentially highly complex alternatives.

In this paper, we develop a goodness-of-fit testing framework for broad classes of semiparametric regression models, including generalised additive models, partially linear models and additive quantile regression models, that can also accommodate applications such as testing treatment effect heterogeneity in observational data.
We work in a setup where we observe independent copies of a random element $D \in \mathcal{D}$, with $X = X(D) \in \mathcal{X}$ some function of the data.
In regression settings---our main focus in this paper---one would have $D = (X,Y)$, with $X$ denoting the covariates and $Y \in \R$ denoting the response.
Given a loss function $\ell:\mathbb{R}\times \mathcal{D}\to \mathbb{R}$ and function $g : \mathcal{X} \to \R$, we write
\[ R(g) := \mathbb{E}[\ell(g(X),D)] \]
for its associated risk.

Many functionals of interest can be identified as minimisers of a certain risk function.
For example, the conditional mean $\mathbb{E}(Y \given X)$ is the minimiser of the quadratic risk corresponding to the loss function $\ell(\eta,d) = (y-\eta)^2$.
We are interested in testing statistical models which constrain such a functional to lie in a structured class.
More precisely, we consider \textit{risk minimiser constraining models}---a class of models which posit that the risk minimiser over a function class $\mathcal{G}$ in fact lies in a smaller class $\mathcal{F}$, where $\mathcal{F} \subset \mathcal{G} \subset \mathbb{R}^{\mathcal{X}}$ are vector spaces of functions.
In other words, the risk minimiser constraining model posits that writing $f^* := \argmin_{f \in \mathcal{F}} R(f)$, we also have that
\begin{equation} \label{eq:null}
f^* \in \argmin_{f \in \mathcal{G}} R(f) .
\end{equation}

For instance, the \textit{additive model} specifies that the conditional mean of the response $Y$ given covariates $X \in \mathbb{R}^d$ is an additive function of the covariates, i.e.\ it lies in $\mathcal{F} := \{f:\mathbb{R}^d \to \mathbb{R} : f(x) = \sum_{j=1}^{d} f_j(x_j)\}$.
This can be therefore viewed as a risk minimiser constraining model by taking $\ell$ to be the squared loss.
For the quantile regression case, we may take $\ell$ to be the appropriate check loss, so if for instance we wish to specify that the conditional median function lies in the class $\mathcal{F}$, we may take $\ell(\eta, d) = |\eta - y|$.
In both of these cases, we may take $\mathcal{G}$ to be the set of all functions $g:\mathcal{X} \to \R$ (with appropriate integrability restrictions) when we are assessing goodness of fit.
However, our framework also allows for $\mathcal{G}$ to be restricted: for example it may contain functions $\{f:\mathbb{R}^d \to \mathbb{R} : f(x) = f_{1,2}(x_1,x_2) + \sum_{j=3}^{d} f_j(x_j)\}$, so interactions in the first two variables are permitted.
The resulting hypothesis test may be interpreted as testing for whether $\mathcal{G}$ offers a significantly better approximation to the ground truth than $\mathcal{F}$.
In our R package \texttt{dScoreTest}, we use this idea to implement a test for comparing nested semiparametric models. 
For the rest of this paper, we will focus on the case of unrestricted $\mathcal{G}$. 

The starting point of our hypothesis testing approach is the alternative characterisation of the null via score equations, which holds under sufficient regularity:
\begin{equation} \label{eq:score_eqns}
\mathbb{E}[\ell'(f^*(X),D) h(X)] = 0, \qquad \forall h \in \G.
\end{equation}
Here $\ell'$ denotes a derivative in the first argument of $\ell$.
In the regression cases where $D=(X, Y)$ and when $\ell$ is squared error loss, this amounts to saying that the population-level residuals $Y - f^*(X)$ are uncorrelated with any random variable of the form $h(X)$, where the function $h \in \mathcal{G}$.
Conversely, if the risk minimiser constraining model is misspecified, then under mild conditions there exists some test function $h \in \mathcal{G}$ for which the expectation above is positive:
thereby exposing the misspecification (see \cref{lem:risklocmin} in \cref{sec:lemmas}).

This characterisation of the null lends itself to what can be described as a `hunt-and-test' approach \citep{rt}: given i.i.d.\ data $(D_i)_{i=1}^{2n}$ we can first use observations $(D_i)_{i=1}^{n}$ to determine an appropriate function $h$ that we hope might expose the signal present in the score under an alternative.
With this function $h$, we can use the remainder $(D_i)_{i=n+1}^{2n}$ of the data to test whether the empirical counterpart of the left-hand side of \eqref{eq:score_eqns} vanishes.
The primary benefit of this sample splitting is that the chosen function $h$ can be treated as fixed during the second, testing stage.
This eliminates the need to account for its sampling variability, allowing the ``hunting'' process to be as elaborate as desired.
However, two key challenges remain: first, it is not obvious how to optimally construct $h$; second, testing the null condition for a given $h \in \mathcal{G}$ is complicated by the fact that $f^*$ is unknown and must be replaced by an estimate $\hat{f}$.

In this work, we introduce the \emph{debiased score test} (DST) to address both of these challenges.
First, we show that for any selected function $h$, we can test this null hypothesis by employing a debiasing step that reduces to a simple weighted least squares regression.
As we shall see, this step renders the test statistic asymptotically insensitive to the estimation error of $\hat{f}$.
Second, we demonstrate that the problem of finding an optimal test function $h$ can be reframed as a specific weighted regression problem.
In this way, the predictive power of modern machine learning methods can be directly leveraged to construct $h$, thereby maximising the statistical power of the resulting test.

The rest of the paper is organised as follows.
After reviewing some related literature, in \cref{sec:dst} we describe our methodology, setting out how hunting and testing may be carried out in general, as well as examining applications to assessing conditional mean and conditional quantile specifications, including testing for heterogeneous treatment effects.
\cref{sec:theory} gives results on uniform asymptotic size control and power at local alternatives.
In \cref{sec:numerical} we present the results of several numerical experiments demonstrating the favourable empirical properties of the procedure.
We conclude with a discussion in \cref{sec:discuss}.
An accompanying R-package \texttt{dScoreTest} available from \url{https://unbiased.co.in/dScoreTest/} provides an implementation of the methodology.
Supplementary material contains the proofs of all the theoretical results and additional details relating to the numerical studies.

\subsection{Related work}  \label{sec:related-work}
There is an extensive body of literature on goodness-of-fit testing for semiparametric regression models such as additive models, and an overview of more classical methods is provided by \citet{gonzalez2013updated}. For instance, 
\citet{dette2001testing} and \citet{gozalo2001testing} propose kernel based tests for additivity; \citet{derbort2002test} propose a test of additivity for lattice data based on the $L_2$ distance between a non-parametric estimate of the regression function and an estimate of the best additive approximation. \citet{sperlich2002nonparametric} and \citet{roca2005testing} propose procedures for testing for interactions in additive models. While foundational, these classical methods face several practical limitations. First, they
typically require using specific classical smoothers and rely on these attaining sufficiently fast rates for estimating $\mathbb{E}(Y \given X = x)$; this may be unrealistic for modern datasets of interest.
Second, they typically rely on using bandwidths that result in undersmoothing, which can be difficult to choose in practice.
Third, to obtain reliable critical values, many of these approaches rely on computationally intensive resampling techniques, such as the bootstrap.

Motivated by the desire for procedures that are relatively easy to use, scale to moderate-dimensional settings, and avoid delicate tuning or bootstrapping, some recent literature has shifted toward methods that can leverage the predictive power of modern machine learning.
For example, \citet{shah2018goodness} and \citet{jankova} develop so-called residual prediction tests  for assessing the goodness of fit of high-dimensional linear and generalised linear models. The methodology
is based on using a machine learning method to extract `left-over signal' from the residuals after a fit of the null model.
\citet{lundborg2024projected} build on this to propose the projected covariance measure to test for the hypothesis of conditional mean independence. \citet{quantile_significance} develop an analogous test for testing the significance of a subset of covariates based on conditional quantile of a response given covariates. Both of these tests can be shown to be special cases of the more general framework we develop here. In particular, they, along with the test of \citet{jankova} employ sample-splitting and a hunt-and-test scheme. This general idea has a long history  \citep[see e.g.][]{moran1973dividing,cox1975note}; recent work \citet{rt} studies this approach for testing the goodness of fit of parametric regression models and also introduces rank-transformed subsampling for aggregating over multiple splits, which we also use in some of our numerical experiments.

Starting from the score equations \eqref{eq:score_eqns}, an alternative hunting for a suitable function $h$ is to use the average of the squares of empirical scores based on a set of pre-specified functions $h_1,\ldots,h_K$. \citet{Sancetta} consider this approach in the context of models defined through reproducing kernel Hilbert spaces. \citet{escanciano} proposes a related strategy for models defined by parametric conditional moment restrictions but instead averages over an infinite set of functions given by a Gaussian process.
The earlier work of \citet{Bickel2006} introduces a very general way of constructing tests for semiparametric models by considering scores orthogonal to the tangent space of scores for the null sub-model, though focuses primarily on parametric models and independence testing.

\citet{williamson2021general} propose a procedure for assessing variable importance which compares the improvement in a specified metric of fit when the variable of interest is included.
While they focus on assessing significance, the approach can be extended naturally to testing the goodness of fit of a regression model as considered in the earlier work of \citet{Yatchew_1992}. \citet{hines2025variable} build on the work of \citet{williamson2021general} to propose a variable importance measure for the conditional treatment effect. We compare our proposal to these approaches in our numerical experiments in \cref{sec:numerical}.

\section{The Debiased Score Test} \label{sec:dst}

Recall that our data are i.i.d.\ copies $(D_i)_{i=1}^{2n}$ of a random element $D$ and $X = X(D) \in \mathcal{X}$ represents some function of the data.
We will use the notation that $L^2(X)$ is the collection of functions $f: \mathcal{X} \to \R$ where $\E f(X)^2 < \infty$.
In regression settings, $D=(X, Y)$ where $X$ are predictors and $Y \in \R$ is a response; unless specified otherwise, we will assume we are in this setting.
We are given a loss function $\ell:\R \times \mathcal{D} \to \R$ that is convex in its first argument.
In the below, we will write $\ell':\R \times \mathcal{D} \to \R$ for a sub-gradient of $\ell( \cdot, d):\R \to \R$ in the sense that
\[ \ell(\eta_2, d) \geq \ell(\eta_1, d) + \ell'(\eta_1,d)(\eta_2 - \eta_1) \]
for all $\eta_1, \eta_2 \in \R$ and $d \in \mathcal{D}$.
(The role that the convexity requirement plays in our proposal will become clear later.) Note that when $\ell$ is differentiable in its first argument, $\ell'$ is simply the derivative; when $\ell(\eta, d) = |\eta - y|$ for example, we can take $\ell'(\eta, d) = \sgn(\eta - y)$ (with $\sgn(0):=0$, say).

Our general strategy for testing the null hypothesis \eqref{eq:null} involves first randomly partitioning our observation indices: $I_1 \cup I_2 := \{1,\ldots,2n\}$ where $|I_1| = |I_2| = n$.
Using data indexed by $I_1$, we construct a test function $\hat{h} : \mathcal{X} \to \R$ designed to expose signal retained in the score under an alternative.
Using data indexed by $I_2$ and the test function $\hat{h}$, we then perform a test.
We describe the test in general terms in the next section, after which we introduce an approach for the hunting stage where $\hat{h}$ is formed.
In \cref{subsec:applications} we discuss
some key applications to testing conditional mean specification, conditional quantile specification and testing for heterogeneity of treatment effects given observational data.

\subsection{Testing} \label{sec:testing}
To introduce our testing approach, suppose, for the moment, that we have access to $f^* = \operatorname{argmin}_{f \in \mathcal{F}} R(f)$.
Setting $L_i = \ell'(f^*(X_i),D_i) \hat{h}(X_i)$, a natural test statistic to consider is the studentised sample average
\begin{equation} \label{eq:oracle_test}
T_n := \frac{\tfrac{1}{\sqrt{n}} \sum_{i=1}^n L_i}{\sqrt{\tfrac{1}{n}\sum_{i=1}^n L_i^2 - \left(\tfrac{1}{n}\sum_{i=1}^n L_i\right)^2}}.
\end{equation}
Under mild conditions, this converges to a standard normal distribution under the null by the central limit theorem.
Thus rejecting for values exceeding the upper $\alpha$ point $z_{1-\alpha}$ of a standard Gaussian yields a test of asymptotic size $\alpha$.

Of course, in practice we do not know $f^*$ and must instead replace this with an estimate $\hat f$.
Using $\hat L_i := \ell'(\hat f(X_i),D_i) h(X_i)$ in \eqref{eq:oracle_test} above introduces an extra term, since
\[ \hat L_i = L_i + R_i, \qquad R_i := \{\ell'(\hat f(X_i),D_i)-\ell'(f^*(X_i),D_i)\} h(X_i). \]
The central limit theorem continues to control the $L_i$ part, but the remainder induces a bias of order $\sqrt{n}\,\mathbb{E}[R_1]$, which typically does not vanish asymptotically when non-parametric methods are used to fit $\hat{f}$ \citep[see e.g.][]{chernozhukov2018dml}, contaminating the limiting distribution of the test.
To characterise the bias more explicitly, let
\begin{equation} \label{eq:phi_w_def}
\phi(\eta,x) := \mathbb{E}[\ell'(\eta,D)\given X=x], \qquad w(x) := \phi'(f^*(x), x) := \partial_\eta \phi(\eta,x)\rvert_{\eta=f^*(x)},
\end{equation}
 and note that a Taylor expansion around $f^*(x)$ gives
\begin{equation} \label{eq:bias1}
\sqrt{n}\,\mathbb{E}[R_1] \;\approx\; \sqrt{n}\,\mathbb{E}\!\left[w(X)\,(\hat{f}(X) - f^*(X))\, \hat{h}(X)\right].
\end{equation}

Note that as $\mathcal{F}$ is a vector space, $\hat{f} - f^* \in \mathcal{F}$.
A key observation then is that if $\hat{h}$ lies in $\mathcal{G}_{\mathrm{DB}}$ given by
\begin{equation} \label{eq:G_debiased}
\mathcal{G}_{\mathrm{DB}} := \left\{ g \in \mathcal{G} : \mathbb{E}\!\left[w(X)g(X)f(X)\right] = 0 \quad \text{for all } f \in \mathcal{F} \right\},
\end{equation}
then the first-order bias given by the right-hand side of \eqref{eq:bias1} would vanish.
As $\ell$ is convex in its first argument, $w(X) \geq 0$ almost surely, so $\mathbb{E}\!\left[w(X)g(X)f(X)\right]$ can be viewed as a weighted inner product between $f(X)$ and $g(X)$.
This observation leads to a convenient way of enforcing the orthogonality condition in \eqref{eq:G_debiased}
via a weighted least squares regression, as the following result indicates.
\begin{proposition} \label{prop:orthog}
    Let $k \leq w(X) \leq K$ for some $0 < k \leq K$.
    Suppose $\mathcal{F} \subset \mathcal{G}$ is a closed linear subspace of $L^2(X)$ and $h \in \mathcal{G}$.
    Then the minimiser $m_{h} \in \mathcal{F}$ of
\[ \E[w(X)\{h(X) - m(X)\}^2] \]
    over $m \in \mathcal{F}$ exists and is unique, and $h - m_{h} \in \mathcal{G}_{\mathrm{DB}}
$. \end{proposition}
This motivates a debiasing strategy that first estimates $m_{\hat{h}}$ through a weighted least squares regression.
Specifically, using some potentially estimated weights $(\hat{w}(X_i))_{i \in I_2}$, we regress the response vector $(\hat{h}(X_i))_{i \in I_2}$ onto the predictors $(X_i)_{i \in I_2}$ using function class $\mathcal{F}$ to produce an estimate $\hat{m}_{\hat{h}}$ of $m_{\hat{h}}$.
For example, we may take
\begin{equation} \label{eq:m_hat}
\hat{m}_{\hat{h}} \in \argmin_{m \in \mathcal{F}} \left( \sum_{i \in I_2} \hat{w}(X_i)\{\hat{h}(X_i) - m(X_i)\}^2 + J(m) \right),
\end{equation}
where $J : \mathcal{F} \to [0, \infty)$ is an appropriate regulariser.
Note however that the particular estimation strategy here is not critical: our theory presented in \cref{sec:theory} only requires $\hat{m}_{\hat{h}}$ to be a good estimate of $m_{\hat{h}}$.
In many settings of interest, for example testing the fit of a generalised additive model, $\phi'$ is known. In this case, we can take $\hat{w}(x)$ as a plug-in estimate given an estimate of $f^*$.

\subsection{Hunting} \label{sec:hunting}
We now turn to the problem of constructing $\hat{h}$, the choice of which determines the power of the procedure.
Consider a local alternative, under which the oracle test statistic $T_n$ \eqref{eq:oracle_test} based on a hunting function $\hat{h} \in \mathcal{G}$ can be expected to be approximately Gaussian with unit variance and mean $\sqrt{n}\,\Lambda(\hat{h})$, where
\begin{equation}
\label{eq:snr} \Lambda(h) := \frac{ \E\!\left[\ell'(f^*(X),D)h(X)\right] }{ \sqrt{ \var\!\left(\ell'(f^*(X),D)h(X)\right) }}
\end{equation}
is the signal-to-noise ratio associated with the direction $h$.
Were $T_n$ to be exactly Gaussian, the power would be maximised for an $\hat{h}$ that maximises $\Lambda$.
The following result suggests a way of finding such a maximising $\hat{h}$.
\begin{proposition} \label{prop:hunt}
Any minimiser of
\begin{equation} \label{eq:recip_wls}
\Gamma(h) := \E \left[ \ell'(f^*(X), D)^2 \left(\frac{1}{\ell'(f^*(X), D)} - h(X) \right)^2\right]
\end{equation}
over $h \in \mathcal{G}$ also maximises $\Lambda(g)$ over $g \in \mathcal{G}$.
Moreover, if a maximiser in the latter problem exists, a minimiser in the former problem exists.
\end{proposition}

Given an estimate $\tilde{f}$ of $f^*$, we can attempt to target an optimal hunting function via the following weighted least squares regression:
\begin{equation} \label{eq:h_hat_star}
\argmin_{h \in \mathcal{G}} \left\{ \sum_{i \in I_1} \ell'(\tilde{f}(X_i), D_i)^2 \left( \frac{1}{\ell'(\tilde{f}(X_i), D_i)} - h(X_i) \right)^2 + J(h) \right\};
\end{equation}
as before, $J$ here is some appropriate regulariser and if $\ell'(\tilde{f}(X_i), D_i)=0$ we simply take its reciprocal to be $1$ (or any finite value---it will not contribute in any case).
Thus for example, when $\ell$ is squared error loss, we regress the reciprocal of the residuals $Y - \tilde{f}(X)$ onto $X$ but with weights given by the square of the residuals.
While using the reciprocal might appear to introduce undesirable instability, this is compensated for by the weights, as can be seen by expanding the square.
Note that any flexible regression method capable of returning functions lying in $\mathcal{G}$ can be used for this task, and this method need not necessarily solve an optimisation of the form in \eqref{eq:h_hat_star}.
In the common case where $\mathcal{G} = L^2(X)$, this can for example be a neural network or the output of a gradient boosting procedure \citep{friedman2001greedy}; we use generalised random forests \citep{grf} in all of our numerical experiments presented in \cref{sec:numerical}.

While the above gives a simple strategy for finding an $\hat{h}$, it does not take account of the fact that prior to testing, the hunted function $\hat{h}$ is to be orthogonalised by replacing it with $\hat{h} - \hat{m}_{\hat{h}}$ (see \cref{sec:testing}); it is not clear that the orthogonalised version continues to estimate an optimal direction.
We describe a refinement of the strategy above that can account for this discrepancy in the next section.

\subsubsection{Accounting for debiasing} \label{sec:accounting_DB}
Ideally, we would like to maximise $\Lambda$, or equivalently minimise $\Gamma$, not over $h \in \mathcal{G}$ but rather over $h \in \mathcal{G}_{\mathrm{DB}}$.
An issue with operationalising this however is that we do not have direct access to this latter function class.
It turns out though that in the common case where $\mathcal{G} = L^2(X)$, there is a way around this.

Let us write $h^*_{\mathrm{DB}}$ and $h^*$ for minimisers (which we assume to exist in the following discussion) of $\Gamma$ over $\mathcal{G}_{\mathrm{DB}}$ and $\mathcal{G}=L^2(X)$ respectively.
Writing $v(X) := \mathbb{E}(\ell'(f^*(X),Y)^2 \given X)$, one can show that these minimisers are related via
\begin{equation}
\label{eq:weighted_distance_to_oracle} h^*_{\mathrm{DB}} \in \argmin_{g \in \mathcal{G}_{\mathrm{DB}}} \mathbb{E}_P\!\left[ v(X) \left\{ h^*(X) - g(X) \right\}^2 \right].
\end{equation}
Geometrically, we can think of $h^*_{\mathrm{DB}}$ as the projection of $h^*$ onto $\mathcal{G}_{\mathrm{DB}}$ under the weighted inner product $\langle f, g \rangle_v := \E[v(X) f(X) g(X)]$.
Writing $\mathrm{proj}_{\mathcal{H}}(h)$ for the projection of $h$ onto a subspace $\mathcal{H}$ under this inner product, we have
\begin{equation}
\label{eq:projection_correction} h^*_{\mathrm{DB}} = \mathrm{proj}_{\mathcal{G}_{\mathrm{DB}}}(h^*) = h^* - \operatorname{proj}_{\mathcal{G}_{\mathrm{DB}}^{\perp}} \left(h^*\right).
\end{equation}
This gives a useful interpretation: the projection of $h^*$ on the orthogonal complement $\mathcal{G}_{\mathrm{DB}}^{\perp}$ of $\mathcal{G}_{\mathrm{DB}}$ (with respect to the inner product $\langle \cdot, \cdot \rangle_v$) is the appropriate correction term to subtract from $h^*$ to yield the optimal direction among the debiased set of functions.
A key observation is that $\mathcal{G}_{\mathrm{DB}}^{\perp}$ can be written in terms of $\mathcal{F}$ (recall that $\mathcal{G}_{\mathrm{DB}}$ is itself defined as an orthogonal complement of $\mathcal{F}$, though with respect to a different weighted inner product \eqref{eq:G_debiased}):
\[ \mathcal{G}_{\mathrm{DB}}^\perp = \frac{w\mathcal{F}}{v} := \left\{ \frac{w(\cdot)f(\cdot)}{v(\cdot)} : f \in \mathcal{F} \right\}. \]
These ideas are illustrated in \cref{fig:optimal_hunting_projection} and form the basis of the following result.
\begin{figure}[t]
\centering
\tdplotsetmaincoords{68}{38}
\begin{tikzpicture}[tdplot_main_coords, scale=1.6, >=stealth]

\coordinate (O) at (0,0,0);
    \coordinate (Rinv) at (1.55,1.25,1.25);
    \coordinate (Horacle) at (1.55,1.25,0);
    \coordinate (Hopt) at (1.55,0,0);

\filldraw[
        fill=blue!35,
        fill opacity=0.16,
        draw=blue!70!black,
        line width=0.8pt
    ]
    (-2.4,-1.65,0) -- (2.8,-1.65,0) -- (2.8,1.85,0) -- (-2.4,1.85,0) -- cycle;

    \node[blue!70!black] at (-1.65,1.45,0) {$\mathcal{G}$};

\draw[
        red!70!black,
        line width=1.2pt
    ]
    (-2.4,0,0) -- (2.8,0,0);

    \node[anchor=west, red!70!black] at (2.92,0,0)
    {$\mathcal{G}_{\mathrm{DB}}$};

\draw[
        black,
        line width=1.05pt
    ]
    (0,-1.65,0) -- (0,1.85,0);

    \node[anchor=east] at (0,-1.85,0)
    {$\mathcal{G}_{\mathrm{DB}}^{\perp}=w\mathcal{F}/v$};

\draw[
        dashed,
        line width=0.9pt
    ]
    (Rinv) -- (Horacle);

\draw[
        dashed,
        line width=0.9pt
    ]
    (Horacle) -- (Hopt);

\draw[
        ->,
        green!55!black,
        line width=1.15pt
    ]
    (O) -- (Horacle);

    \draw[
        ->,
        green!55!black,
        line width=1.15pt
    ]
    (O) -- (Hopt);

\filldraw[black] (Rinv) circle (1.2pt);
    \filldraw[black] (Horacle) circle (1.2pt);
    \filldraw[black] (Hopt) circle (1.2pt);

\node[anchor=south west] at (Rinv)
    {$\ell'^{-1}$};

    \node[anchor=west] at (Horacle)
    {$h^*$};

    \node[anchor=north] at (Hopt)
    {$h^*_{\mathrm{DB}}$};

\draw[line width=0.7pt]
    (1.36,0,0) -- (1.36,0.18,0) -- (1.55,0.18,0);

\end{tikzpicture}
\caption{Geometric interpretation of our optimal hunting strategy that accounts for debiasing. The point
$\ell'^{-1} := \ell'^{-1}(f^*(X), D)$ is first projected onto $\mathcal{G}$ to obtain
$h^*$. The component in
$\mathcal{G}_{\mathrm{DB}}^{\perp}=w\mathcal{F}/v$ is then removed,
giving $h^*_{\mathrm{DB}}\in\mathcal{G}_{\mathrm{DB}}$.}
\label{fig:optimal_hunting_projection}
\end{figure}
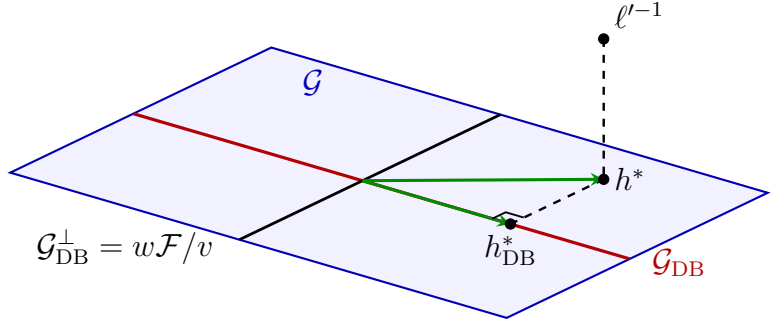

\begin{proposition} \label{prop:opt_hunt}
    Under regularity conditions stated in \cref{sec:regularity_hunt}, we have that both $h^*$ and $h^*_{\mathrm{DB}}$ exist and are unique.
    Moreover
\begin{equation} \label{eq:h_cor}
h^* - h^*_{\mathrm{DB}} = \operatorname{proj}_{\mathcal{G}_{\mathrm{DB}}^{\perp}} \left(h^*\right) = \frac{w}{v} \argmin_{f \in \mathcal{F}} \mathbb{E}\!\left[ \frac{w(X)^2}{v(X)} \left( \frac{v(X)h^*(X)}{w(X)} - f(X) \right)^2 \right].
\end{equation}
\end{proposition}
This motivates the following hunting procedure.

\begin{enumerate}
\item Form $\hat{h}^*$ via the weighted regression in \eqref{eq:h_hat_star} using any user-chosen regression method.
\item Let $\hat{v}$ be an estimate of the conditional variance $v$ formed, for example, through regressing $\ell'(\tilde{f}(X), D)^2$ onto $X$ (again, using any suitable method; in our experiments, we use random forests).
\item Given weight function estimate $\tilde{w}$ (recall \eqref{eq:m_hat} and the following discussion), estimate the correction to $h^*$ (i.e.\ the left-hand side of \eqref{eq:h_cor}) through a weighted least squares regression of response $\hat{v}(X) \hat{h}^*(X) / \tilde{w}(X)$ onto $X$ with weights $\tilde{w}(X)^2/\hat{v}(X)$ over the original null function class $\mathcal{F}$. For example, we may take
\[ \hat{f}_{\mathrm{proj}} := \argmin_{f \in \mathcal{F}} \left[ \sum_{i \in I_1} \left\{ \frac{\tilde{w}(X_i)^2}{\hat{v}(X_i)} \left(\frac{\hat{v}(X_i) \hat{h}^*(X_i)}{\tilde{w}(X_i)} - f(X_i) \right)^2\right\} + J(f) \right]. \]

\item Finally set $\hat{h} = \hat{h}^* - \tilde{w} \hat{f}_{\mathrm{proj}} / \hat{v}$.
\end{enumerate}

\begin{algorithm}[!htb] \raggedright  \caption{Debiased Score Test (In Sample)} \label{alg:dst_in}
	\textbf{Input}: Data $(D_i)_{1 \leq i \leq 2n}$. Loss function $\ell : \mathbb{R} \times \mathcal{D} \to \mathbb{R}$. Linear function classes $\F \subset \G$. 
	\vskip .3em
	\begin{algorithmic}
		\State Split sample using the index sets $\{1,\dots, 2n\} = I_1 \cup I_2$ with $|I_1| = n$, $|I_2| = n$.
        
		\State \textit{Hunt}: On sample $I_1$, determine $\hat{h}$ as follows:
        \State \hskip0.5em (a) Fit the null model to obtain estimate $\tilde{f}$ of $f^*$.
        \State \hskip0.5em (b) Regress the $\ell'(\tilde{f}(X_i), D_i)^{-1}$ onto the $X_i$ with observation weights $\ell'(\tilde{f}(X_i), D_i)^2$ using a flexible regression method of choice. 
        \State \hskip0.5em (c) (Optional) In the case where $\mathcal{G} = L^2(X)$, compute the refinement in \cref{sec:accounting_DB}.

\State On sample $I_2$, fit $\hat{f}\in \F$ to approximate $f^{*}$.
            
		\State \textit{Debias} : On sample $I_2$, fit $\hat{m}_{\hat{h}}$ to approximate $m_{\hat{h}}$ by regressing $\hat{h}$ on $\mathcal{F}$ with weights $w_i = \hat{w}(X_i)$ for $i \in I_2$.
        
		\State \textit{Test}: For $i \in I_2$ compute $L_i = \ell'(\hat{f}(X_i),D_i)\{\hat{h}(X_i)-\hat{m}_{\hat{h}}(X_i)\}$ and compute the test statistic
		\begin{equation*}
	T_n := \frac{\frac{1}{\sqrt{n}}\sum_{i=1}^n  L_i}{\sqrt{\frac{1}{n}\sum_{i=1}^n L_i^2 - \left(\frac{1}{n} \sum_{i=1}^n L_i\right)^2}} 
		\end{equation*}
	
	\end{algorithmic}
\end{algorithm}

The full DST algorithm is summarised in \cref{alg:dst_in} (we describe this as `in sample' to contrast with \cref{alg:dst_out} that introduces an additional sample split for theoretical tractability).
Note that although the construction above is designed to produce a hunting function that already lies in $\mathcal{G}_{\DB}$, we have found empirically that the debiasing step remains helpful for maintaining Type I error control.
In practice, we suggest performing both the refinement above and the debiasing step; see also \cref{fig:qgam-power} for numerical evidence.

\subsection{Multiple sample splitting} \label{sec:multi_split}
The data splitting in the DST, while necessary to produce an approximately standard Gaussian final test statistic under the null, as with other similar randomised tests can hamper reproducibility and potentially reduce power compared to an approach that attempts to average out this randomness. For these reasons, in practice it is advisable to compute the DST multiple times with different random splits and average the resulting test statistics. Though composed of marginally Gaussian test statistics, this aggregate test statistic is typically not Gaussian under the null, as the dependence across the different splits can be highly complex. To construct a calibrated $p$-value from the final test statistic, we recommend using rank-transformed subsampling \citep{rt}, which uses subsampling to effectively estimate the dependence structure. We illustrate the use of this with the DST in our numerical experiments in \cref{sec:numerical}.

\section{Applications} \label{subsec:applications}
We now present some examples where the DST may be used to test model specification, so in particular we take $\mathcal{G} = L^2(X)$.
To do this, we will cast each model as a risk-constraining model by specifying the loss function $\ell$ and the function class $\mathcal{F}$. We focus on regression settings, so $D = (X, Y) \in \mathcal{X} \times \R$.
We will identify the score function $\ell'(\eta,D)$, the smoothed score function $\phi(\eta,X)$ and the weighting function $w(X)$, which may then be used in \cref{alg:dst_in} to yield the testing procedure.

\subsection{Testing conditional mean specification} \label{sec:cond-mean}
Many parametric and semiparametric statistical models impose that the conditional mean of the response given covariates takes the form
\[
\E(Y \given X) = \mu(f^*(X)),
\]
where $\mu$ is a known differentiable strictly increasing inverse link function and $f^* \in \mathcal{F}$, a vector space of functions.
By choosing $\mathcal{F}$ appropriately, this framework accommodates testing a wide variety of commonly used models, including classical \mbox{(quasi-)Generalised Linear Models}\footnote{We explain in \cref{sec:GLM} how the DST for this simple case relates to the classical score test for variable significance in generalised linear models.} \citep{glm,wedderburn1974quasi},
Generalised Additive Models \citep{gam}, Partially Linear Models \citep{plm}, and Varying Coefficient Models \citep{vcm}. For example, the first of these are recovered by
taking $\mathcal{F}$ to be functions of the form $x \mapsto x^\top\beta + c$ for some $\beta \in \mathbb{R}^d, c \in \mathbb{R}$.

To see how these may be viewed as risk-constraining models, define the loss function
\[ \ell(\eta, d) := -\eta y + K(\eta), \]
where $K$ is an antiderivative of $\mu$.
We show in \cref{lem:loss-gam} (see \cref{sec:lemmas}) that $\eta \mapsto \ell(\eta,d)$ is convex and that the function $f^*$ minimises the corresponding risk over $L^2(X)$.
We have $\ell'(\eta, D) = \mu(\eta) - Y$, $\phi(\eta, X) = \mu(\eta) - \mathbb{E}(Y \given X)$, and $w(X) = \mu'(f^*(X))$.
When $\mu$ is the identity, $X=(T, Z)$ and $\mathcal{F}$ is the collection of functions of $z$, the corresponding null reduces to $T$ and $Y$ being mean independent conditional on $Z$ and the DST in this case recovers the projected covariance measure of \citet{lundborg2024projected}.

\subsubsection{Detecting treatment effect heterogeneity} \label{sec:hte}
One case of testing conditional mean specification that is of interest concerns detecting treatment effect heterogeneity.
Consider a setting with real-valued treatment $T$, real-valued outcome $Y$ and baseline covariates $Z \in \mathbb{R}^{d}$.
Let $Y(t)$ be the potential outcome under treatment $t$.
Let $\tau(z,t) := \E[Y(t) - Y(0) \given Z=z]$ be the conditional average treatment effect (CATE) function, which is identified as $\tau(z,t) = \E[Y \given T=t, Z=z] - \E[Y \given T=0, Z=z]$ under standard assumptions including positivity and no unmeasured confounding.
Given a subset of covariates $S \subset \{1,\dots,d\}$, we are interested in assessing the treatment effect heterogeneity in the strata defined by $Z_S$.
Nonparametrically, we consider the null hypothesis of no heterogeneity in $Z_{S}$ while holding all the other covariates $Z_{S^{c}}$ fixed, which posits that the CATE $\tau(z,t)$ only depends on $z$ through $z_{S^{c}}$.
Equivalently, the null hypothesis posits the conditional mean specification
\[ \mathbb{E}(Y \given T = t, Z = z) = \mu_0(z) + \mu_1(t,z_{S^c}) \]
for functions $\mu_0$ and $\mu_1$.

In what follows, we focus on the case when $T$ is binary, where we can always write $\mathbb{E}(Y \given T=t, Z = z) =: \mu(t,z) := \mu_0(z) + t \cdot \tau(z)$ with $\tau(z) := \tau(z,1)$.
In this case, the null hypothesis amounts to positing the following function class for $\mu(t,z)$:
\[ \mathcal{F} := \{\mu(t,z) = \mu_0(z) + t \cdot \tau(z): \mu_0 \in \mathcal{F}_0, \tau \in \mathcal{F}_S\}, \]
where $\mathcal{F}_0$ is unrestricted while $\mathcal{F}_{S}$ is given by
\[ \mathcal{F}_{S} := \{\tau: \text{$\tau(z)$ only depends on $z$ through $z_{S^{c}}$} \}. \]

We first devise an algorithm for fitting the null model using a machine learning algorithm of choice.
Defining $e(Z) := \E[T \given Z]$ and $\mu(Z) := \E[Y \given Z]$,
we have the following decomposition \citep{robinson1988root}:
\begin{equation} \label{eq:R-decomp}
Y - \mu_0(Z) - T \cdot \tau(Z) = \underbrace{\left\{\mu(Z) - \mu_0(Z) - e(Z) \tau(Z) \right\}}_{\mathrm{(I)}} + \underbrace{\left\{ Y - \mu(Z) - (T - e(Z))\tau(Z) \right\}}_{\mathrm{(II)}}.
\end{equation}
Here (I) is a function of $Z$, while (II) has zero mean conditioned on $Z$.
It follows that
\[ \E \{Y - \mu_0(Z) - T \cdot \tau(Z)\}^{2} = \E \mathrm{(I)}^{2} + \E \mathrm{(II)}^{2}. \]
To minimise the risk on the left-hand side over $\mathcal{F}$, note that $\mu_0$ can be chosen to make (I) vanish pointwise, upon which $\tau$ becomes the minimiser of $\E \mathrm{(II)}^{2}$ within $\mathcal{F}_S$.
Having found this minimiser $\tau^{\star}$, we then know $\mu_0^{\star}$ is the minimiser of $\E \mathrm{(I)}^{2}$, or equivalently, $\E \{Y - T \cdot \tau^{\star}(Z) - \mu_0(Z) \}^{2}$.
This motivates \cref{alg:fit_null_dml} for fitting the null model, which estimates $\tau$ using the loss $\E \mathrm{(II)}^{2}$. 
As commonly implemented in the literature \citep{nie2021quasi,foster2023orthogonal}, we use cross-fitting to form a Neyman-orthogonal R-loss.

We now turn to the debiasing step. Suppose the hunted function is $\hat{h}(t,z) = \hat{h}_0(z) + t \cdot \hat{h}_{\Delta}(z)$ with $\hat{h}_{\Delta} \notin \mathcal{F}_{S}$ in general; our goal is to estimate its projection $m_{\hat{h}}(t,z) =: m_0(z) + t \cdot m_{\Delta}(z_{S^{c}}) $ onto the null model.
Recall that by \cref{prop:orthog} and \eqref{eq:weighted_distance_to_oracle} we have
\[ (m_0, m_{\Delta}) = \argmin_{f_0 \in \mathcal{F}_0, f_{\Delta} \in \mathcal{F}_{S}} \E \left\{\hat{h}_0(Z) - f_0(Z) + T \cdot (\hat{h}_{\Delta}(Z) - f_{\Delta}( Z_{S^{c}}) ) \right\}^{2}, \]
where $\hat{h}$ is treated as a fixed function.
By a decomposition similar to \eqref{eq:R-decomp}, we see that
\[ m_{\Delta} = \argmin_{f_{\Delta} = f_{\Delta}(Z_{S^{c}})} \E\{(T-e(Z))^{2}\,(\hat{h}_{\Delta}(Z) - f_{\Delta}(Z_{S^{c}}))^{2} \}, \, m_0(z) = \hat{h}_0(z) + e(z) (\hat{h}_{\Delta}(z) - m_{\Delta}(z_{S^{c}})). \]
Therefore, we may estimate $m_{\Delta}$ through a weighted least squares regression onto $Z_{S^c}$ using a machine learning method of choice, to give $\hat{m}_{\Delta}$, and then obtain the debiased hunted function 
\[ (\hat{h} - \hat{m}_{\hat{h}})(t,z) := (t - \hat{e}(z))\,(\hat{h}_{\Delta}(z) - \hat{m}_{\Delta}(z_{S^{c}})), \]
where $\hat{e}(z)$ is an estimate of the propensity $e(z)$.

\begin{algorithm}[htb]
\caption{Fitting the null CATE model}
\label{alg:fit_null_dml}
\begin{algorithmic}[1]
\State Partition $\{1,\dots,n\}$ into $K$ (5 by default) folds $J_1,\dots,J_K$. Let $J_{-k} := \cup_{j \neq k} J_j$.
\For{$k = 1,\ldots,K$}
    \State Use data in $J_{-k}$ to fit $\hat{e}_{-k}(z)$ and $\hat{\mu}_{-k}(z)$.
    \State Compute $\tilde{Y}_i := Y_i - \hat{\mu}_{-k}(Z_i)$ and $\tilde{T}_i:=T_i - \hat{e}_{-k}(Z_i)$ for each $i \in J_{k}$. 
\EndFor
\State Fit $\hat{\tau} = \argmin_{\tau} \sum_{i=1}^{n} (\tilde{Y}_i - \tilde{T}_i \cdot \tau (Z_{i,S^{c}}))^{2} =  \argmin_{\tau} \sum_{i=1}^{n} \tilde{T}_i^{2}\, (\tilde{Y}_i / \tilde{T}_i - \tau (Z_{i,S^{c}}))^{2}$ using a flexible regression method.
\State Fit $\hat{\mu}_0 = \argmin_{\mu_0} \sum_{i=1}^{n} (Y_i - T_i \cdot \hat{\tau}(Z_{i,S^{c}}) - \mu_0(Z_i))^{2}$ using flexible regression.
\State \Return $(\hat{\mu}_0, \hat{\tau})$. 
\end{algorithmic}
\end{algorithm}

\subsection{Testing conditional quantile specification} \label{sec:quant}
Fix $\tau \in (0,1)$ and consider
a model which posits that the $\tau$-conditional quantile $q_\tau(x) := \inf \left\{ y \in \mathbb{R} : \mathbb{P}(Y \leq y \given X = x) \geq \tau \right\}$ of $Y$ given $X$ lies in the vector space of functions $\mathcal{F}$, such as the class of additive functions.

It is well-known that $q_{\tau}(\cdot)$ is the minimiser
of the risk $R(g)$ associated with the check loss
\[ \ell(\eta, d) := \rho_{\tau}(y - \eta), \quad \text{where} \quad \rho_{\tau}(r) := \tau r - r \ind_{\{r \leq 0\}}. \]
(See \cref{lem:conditional_quantile} in \cref{sec:lemmas} for a formal statement.)
Such a model is therefore a risk-constraining model with $\ell'(\eta, D) = \ind_{\{Y \leq \eta\}} - \tau$, $\phi(\eta, X) = \mathbb{P}(Y \leq \eta \given X) - \tau$, and $w(X) = p_{Y|X}(f^*(X) \given X)$, where $p_{Y|X}$ is the conditional density of $Y$ given $X$, assumed to exist. We describe a convenient way of estimating $w$ in \cref{sec:qgam-simu}.

\section{Theory} \label{sec:theory}
In this section, we provide a general result on the uniform asymptotic properties of the DST (\cref{thm:asymptotic}) from which may be derived results on Type I error control (\cref{cor:type_1_error}) and power (\cref{cor:power}). In \cref{sec:theory_examples} we provide more concrete and interpretable sufficient conditions for these general results to hold in the cases of testing conditional mean specification and testing conditional quantile specification.

Uniform, as opposed to pointwise asymptotic results, such as we provide here, are particularly important in semiparametric settings. Indeed, considering the problem of testing conditional mean independence of a response and a chosen predictor given remaining continuous predictors, \citet{Shah_Peters_2020} shows in particular that this problem is fundamentally impossible without imposing additional conditions. In contrast to pointwise analyses, uniform results more clearly reveal these conditions and also provide genuine asymptotic guarantees on Type I error.

In order to present these uniform results, given a distribution $P$ for $D$, we will often subscript associated quantities such as expectations (including risks) and probabilities by $P$.
	As in \citet{Shah_Peters_2020}, given a family of sequences of real-valued random variables $(U_{P,n})_{P\in\mathcal{P}, n\in\mathbb{N}}$ whose distributions are determined by $P\in\mathcal{P}$ and a deterministic sequence $(r_n)_{n \in \mathbb{N}}$, we write
$U_{P,n} = o_{\mathcal{P}}(r_n)$ if $\sup_{P\in\mathcal{P}}\pr_P(|U_{P,n}/r_n| > \epsilon) \to 0$ for every $\epsilon > 0$.

\subsection{Regularity conditions} \label{subsec:setup}
Our results require the following mild regularity conditions relating to a class of distributions $\mathcal{P}$ over which our results to be stated will be uniform. First let us define $\mathcal{F}_P$ to be the closure of $\mathcal{F} \cap L_P^2(X)$ in $L_P^2(X)$ (identifying functions which are almost surely equal under $P$), and define $\mathcal{G}_P$ similarly.

The following assumption in particular ensures that the risk $R_P$ is well-defined.
\begin{assumption}[Loss and score] \label{assump:risk_exist} 
For each $P \in \mathcal{P}$, suppose $\mathbb{E}_P|\ell(g(X),D)| $ and $\mathbb{E}_P\, \ell'(g(X),D)^2 $ are finite $\text{for each } g \in \mathcal{G}_P$.
\end{assumption}

\begin{assumption}[Existence of risk minimiser] \label{assump:minimiser-existence}
There exists $f^*_P \in \argmin_{f \in \mathcal{F}_P}R_P(f)$ for each $P \in \mathcal{P}$.
\end{assumption}
With this, the class of null distributions may be written as
\begin{equation} \label{eq:def-null}
\mathcal{P}_0 := \left\{P \in \mathcal{P}: f^*_P \in \argmin_{g \in \G_P} R_{P}(g) \right\}.
\end{equation}

\begin{assumption}[Regularity conditions of risk]\label{assump:risk-regularity}
  There exist universal constants $L,L',k$ and $ K > 0$ such that the following hold.
    \begin{assumpenum}
        \item \label{assump4_risk-lipschitz} There exists $\gamma \in [0,1]$ such that
\[ \sup_{P \in \mathcal{P}}\mathbb{E}_P[\{\ell'(\eta,D)-\ell'(\zeta,D)\}^2 \given X = x] \leq L\,|\eta -\zeta|^{1+\gamma}, \quad \forall \eta, \zeta \in \mathbb{R} \text{ and } x \in \mathcal{X} . \]
\item \label{assump4_bdd-weights} The smoothed score $\eta \mapsto \phi_P(\eta,x)$ \eqref{eq:phi_w_def} is differentiable everywhere for each $P \in \mathcal{P}, x \in \mathcal{X}$ and $\sup_{x \in \mathcal{X},\eta \in \mathbb{R}}|\phi_P'(\eta,x)| \leq K$ for each $P \in \mathcal{P}$.
        \item \label{assump4_lower_bdd-weights} The weight function $w_P(x) \geq k$ for all $x \in \mathcal{X}$ and $P \in \mathcal{P}$.

        \item \label{assump4_lipschitz-weights} The derivative of the smoothed score $\phi'_P(\cdot, x): \mathbb{R} \to \mathbb{R}$ is $L'$-Lipschitz for each $x \in \mathcal{X}$ and $P \in \mathcal{P}$.
        
    \end{assumpenum}
\end{assumption}

Assumption~\ref{assump4_bdd-weights} can be shown to hold for the loss functions we consider in \cref{subsec:applications} under mild conditions. For these settings, one can show that Assumption~\ref{assump4_bdd-weights} implies Assumption~\ref{assump4_risk-lipschitz}. 
Note that Assumption~\ref{assump4_bdd-weights} implies that $w_P(x) \leq K$ for each $x \in \mathcal{X}$ and $P \in \mathcal{P}$, and so together with Assumption~\ref{assump4_lower_bdd-weights}, \cref{lem:proj-existence} ensures that the projection $m_{\hat{h}}$ exists and is unique.
Assumption~\ref{assump4_lipschitz-weights} allows us to control the remainder term that arises from using a first-order approximation to the bias term in \eqref{eq:bias1}.

\begin{assumption}[Score moments] \label{assump:score_variance}
There exist $0 < c \leq C$ such that for each $P \in \mathcal{P}$ we have that $c \leq \var_P(\ell'(f^*_P(X),D) \given X = x)$ and $\E_P(\ell'(f^*_P(X),D)^2 \given X = x) \leq C$ for all $x \in \mathcal{X}$.
\end{assumption}
This assumption in particular is used to ensure that the denominator of the test statistic is sufficiently well-behaved.

\subsection{Main results} \label{sec:main_results}
In the rest of this section, we present general conditions ensuring uniform asymptotic validity and power of our test, primarily by imposing assumptions on the performance of the estimation procedures involved.
To study the theoretical properties of the DST, it is convenient to work with a version that uses three, rather than two, splits of the data, as per \cref{alg:dst_out} below. The extra split creates additional independence that simplifies the analysis. However, its role is very different from the more fundamental split between the hunting and testing samples in \cref{alg:dst_in}. In particular, one can eliminate the efficiency loss due to the additional sample split in \cref{alg:dst_out} through a simple cross-fitting strategy that involves exchanging the roles of $I_2$ and $I_3$ while this is not true for the hunt and test split (see also \cref{sec:multi_split}). 
However, in practice, we recommend using the simpler two-split version stated in \cref{alg:dst_in}.

\begin{algorithm}[!htb] \raggedright  \caption{Debiased score test (three-split version)} \label{alg:dst_out}
\textbf{Input}: Data $\{D_1,\dots,D_{3n}\}$, with other inputs as in \cref{alg:dst_in}.
        \vskip .3em
	\begin{algorithmic}[1]
\State Split sample into $\{1,\dots, 3n\} = I_1 \cup I_2 \cup I_3 $ with $|I_1| = n$, $|I_2| = n$ and $|I_3| = n$.
\State  Perform \cref{alg:dst_in} but use $I_1$ to hunt, $I_2$ to learn $\hat{f}$ and $\hat{m}_{\hat{h}}$ and $I_3$ to compute the test statistic.
    \end{algorithmic}
\end{algorithm} 

To reduce clutter, we use shorthand for the following random variables:
\begin{equation} \label{eq:eps}
\varepsilon_P:= \ell'(f^*_P(X),D), \qquad \xi_P:= \hat{h}(X)-{m}_{\hat{h}}(X), \qquad \sigma^2_P := \mathbb{E}_P(\xi_P^2 \given \hat{h}).
\end{equation}
In the case of testing conditional mean specification, $\varepsilon_P$ is the raw population-level residual $Y - \mu(f^*_P(X))$. The quantity $\sigma_P^2$ can be thought of as the noise variance in the (weighted) regression of $\hat{h}(X)$ onto $X$ over the function class $\mathcal{F}$, conditional on $\hat{h}$.
Note that by conditioning on a function such as $\hat{h}$, formally we mean conditioning on all the data used to form it: in the case of $\hat{h}$, this would be the sample $I_1$ as well as any external randomness involved in its construction.
The following measures of the predictive performance of the regression procedures used to estimate $f^*_P$ and $m_{\hat{h}}$ play a critical role in the asymptotic behaviour of the test statistic:
\begin{align*}
\mathcal{E}_1 &:= \mathbb{E}_P(\{\hat{f}(X) - f^*_P(X)\}^2 \given\hat{f}), \\
\mathcal{E}_2 &:= \frac{1}{\sigma^2_P}\mathbb{E}_P(\{\hat{m}_{\hat{h}}(X) - m_{\hat{h}}(X)\}^2 \given \hat{m}_{\hat{h}}, \hat{h}), \\
\mathcal{E}_{3} &:= \mathbb{E}_P\left( \left.\frac{1}{\sigma^2_P}\{\hat{f}(X) - f^*_P(X)\}^2 \xi_P^2 \,\right|\, \hat{f}, \hat{h}\right ).
\end{align*}
The error term $\mathcal{E}_2$ is a scaled mean squared error. The rationale for this scaling factor is that under the null, $\hat{h}$ and hence also $\xi_P$ is effectively an estimate of the zero function, as there is nothing to hunt for. Hence one would expect the unscaled error $\mathbb{E}_P(\{\hat{m}_{\hat{h}}(X) - m_{\hat{h}}(X)\}^2 \given \hat{m}_{\hat{h}}, \hat{h})$ to be very small. Dividing by the noise variance here can be thought of as bringing the error up to the appropriate scale. The error term $\mathcal{E}_3$ is a weighted version of $\mathcal{E}_1$. To gain some intuition for this quantity, observe that were $|\xi_P| \leq \beta\, \sigma_P$ for some $\beta > 0$ almost surely, then $\mathcal{E}_3 \leq \beta^2 \mathcal{E}_1$. 
\begin{assumption}[Estimation rates]\label{assump:mse}
We have that $\mathcal{E}_2 = o_{\mathcal{P}}(1), \mathcal{E}_3 = o_{\mathcal{P}}(1) $ and the product MSPE satisfies
\[ \sqrt{\mathcal{E}_1 \mathcal{E}_2} = o_{\mathcal{P}}(n^{-1/2}). \]
If the map $\phi_P'(\cdot,x)$ is non-constant for some $x \in \mathcal{X}$, then we also require
\[ \sqrt{\mathcal{E}_1 \mathcal{E}_3} = o_{\mathcal{P}}(n^{-1/2}). \]
\end{assumption}
The rate requirements here are non-trivial but also standard in the double/debiased machine learning literature; see, e.g., \citet{chernozhukov2018dml,rotnitzky2021characterization,foster2023orthogonal}. In particular, they are slow enough to accommodate many semiparametric procedures. It is worth noting that in the case of conditional mean specification for example, \cref{assump:mse} places no rate requirement on nonparametric estimation of the regression function or the quality of the hunting procedure; the latter however impacts the power as shown in \cref{cor:power} to follow.

In order to state our main asymptotic result concerning our test, we first define for each $P \in \mathcal{P}$,
\[ s_P := \argmin_{g \in \mathcal{G}_P }\E_P \{g(X)-\ell'(f^*_P(X),D)\}^2, \qquad \text{and} \qquad  \tau_P := \mathbb{E}_P\,s^{2}_P(X).\]
These are well-defined by \cref{lem:alt_proj} in \cref{sec:lemmas}. The quantity $\tau_P$ is zero under the null, since then $\ell'(f^*_P(X),D)$ is uncorrelated with any $g(X)$ with $g \in \mathcal{G}_P$ \eqref{eq:score_eqns}, but will be positive under an alternative. It may thus be thought of as a measure of the effect size.
We first state a general result that characterises the asymptotic behaviour of the test before discussing the consequences for Type I error and power.

\begin{theorem}[Asymptotic behaviour of the test] \label{thm:asymptotic}
Suppose we have independent observations $(D_i)_{1\leq i \leq 3n}$ from a distribution $P \in \mathcal{P}$ and \crefrange{assump:risk_exist}{assump:mse} hold. Suppose we use \cref{alg:dst_out} to form the test statistic $T_n$.
Additionally assume that
    \begin{enumerate}
         \item $\sup_{P \in \mathcal{P}} \mathbb{P}_P(\sigma^2_P = 0) = o(1)$, and \label{thm1:assump_a}
         \item there exists $\delta \in (0,2]$ such that $\mathbb{E}_P(|\varepsilon_P \xi_P|^{2 + \delta} \given \hat{h})/{\sigma^{2+\delta}_P} = o_{\mathcal{P}} (n^{\delta/2}).$ \label{thm1:assump_b}
     \end{enumerate}
    Then, it holds that
\[ T_n = Z_n + \sqrt{n}\, (1+ r_n) \,\Lambda_{n,P}, \]
where $\Lambda_{n,P} := \Lambda_P(\hat{h}-{m}_{\hat{h}})$ for $\Lambda_P$ defined as in \eqref{eq:snr}, $r_n = o_{\mathcal{P}}(1)$ and
\[ \sup_{P \in \mathcal{P}} \sup_{t \in \mathbb{R}} \left|\mathbb{P}_P(Z_n \leq t) - \Phi(t)\right| \to 0, \]
where $\Phi$ is the cumulative distribution function of a standard Gaussian.
\end{theorem}
Condition~\ref{thm1:assump_a} is a mild condition that ensures non-degeneracy of the test statistic. Condition~\ref{thm1:assump_b} is a (conditional) Lyapunov condition used to apply a central limit theorem. These assumptions are also used in \citet{lundborg2024projected}. The result shows that the DST test statistic $T_n$ is asymptotically the sum of a standard Gaussian and a certain shift involving $\Lambda_{n,P} = \E_P(\varepsilon_P \xi_P) / \sqrt{\var(\varepsilon_P \xi_P)}$. Since under the null, we have  $\E_P(\varepsilon_P \xi_P) = 0$, we immediately obtain
the following result on Type I error control.

\begin{corollary}[Type I error control] \label{cor:type_1_error}
Consider the setup of \cref{thm:asymptotic} but with $\mathcal{P}$ set to the class of null distributions $\mathcal{P}_0$ \eqref{eq:def-null}.
Then, we have that
\[ \sup_{P \in \mathcal{P}_{0}} \sup_{t \in \mathbb{R}} |\mathbb{P}(T_n \leq t) - \Phi(t)| \to 0. \]
\end{corollary}

Let us now turn to the power under a sequence of local alternatives.
For $P \in \mathcal{P}$, by \cref{lem:snr_bound} in \cref{sec:lemmas}, one has that $\Lambda_{n,P} \geq A\tau_P^{1/2}\operatorname{Corr}_P(s_P(X),\xi_P)$ for a positive constant $A$.
Thus, as long as the hunted signal $\xi_P$ is sufficiently correlated with $s_P(X)$, by \cref{thm:asymptotic} we can expect the test to be powerful.

\begin{corollary}[Asymptotic power]\label{cor:power}
Consider the setup of \cref{thm:asymptotic} and suppose sequences $(\epsilon_n)$ and $(\rho_n)$ are such that
\begin{equation} \label{eq:quality_of_hunting}
\inf_{P \in \mathcal{P}(\epsilon_n)}\mathbb{P}_P(\operatorname{Corr}_P(s_P(X),\xi_P \given\hat{h}) > \rho_n) \to 1
\end{equation}
and $n\rho^2_n\epsilon_n \to \infty$,
where $\mathcal{P}(\epsilon) := \{P \in \mathcal{P} : \tau_P \geq \epsilon\}$. Then for any $\alpha \in (0,1)$,
$$ \inf \limits_{P \in \mathcal{P}(\epsilon_n)} \mathbb{P}_P(T_n>z_{1-\alpha}) \to 1. $$
\end{corollary}
In particular, given a sequence of local alternatives for which $n \epsilon_n \to \infty$, if  \eqref{eq:quality_of_hunting} holds with $\rho_n = \rho$ for some $\rho > 0$, then the DST asymptotically has power.
In our simulations, we have often found this to be the case.

\subsection{Examples} \label{sec:theory_examples}
Here we specialise the theory above to the applications discussed in \cref{subsec:applications}. Throughout, we shall assume the risk minimiser $f^*_P$ on $\mathcal{F}_P$ exists for each distribution of interest $P \in \mathcal{P}$ (\cref{assump:minimiser-existence}). We also assume the estimation rates in \cref{assump:mse} to hold uniformly on $\mathcal{P}$.

\subsubsection{Conditional mean specification}
\label{subsubsec:cmean-theory}
Recall that for the problem of testing conditional mean specification, the loss takes the form $\ell(\eta, d):= -\eta y + K(\eta)$, where $K$ is an antiderivative of the increasing link function $\mu$. Below we provide a list of conditions under which the more abstract regularity conditions in \cref{assump:risk_exist,assump:risk-regularity,assump:score_variance} hold (see \cref{lem:conditional_mean}).

\begin{assumption} \label{assump:mean}
Suppose that the link function $\mu$ and class of distributions $\mathcal{P}$ on $\mathbb{R}^d \times \mathbb{R} $ satisfy:
\begin{assumpenum}
    \item The link function $\mu : \mathbb{R} \to \mathbb{R}$ is increasing, with $\|\mu'\|_{\infty} \leq L$ and $\|\mu''\|_{\infty} \le L'$ for some $L,L' \geq 0$.
    \label{assump6_bdd_link}

    \item There exists a constant $k > 0$ such that $\mu'(f^*_P(x)) > k$ for each $x \in \mathcal{X}$ and $P \in \mathcal{P}$.
    \label{assump6_bdd_weights}

    \item There exist $0 < c \leq C$ such that $c \leq \var_P(Y \given X = x) \leq C$ for each $x \in \mathcal{X}$ and $P \in \mathcal{P}$.
    \label{assump6_bdd_variance}

\item There exists $M \geq 0$ such that $\sup_{x \in \mathcal{X}} |\mu(f^*_P(x))-\mu(g^*_P(x))| \leq M$ for each $P \in \mathcal{P}$.
\label{assump6:bdd_alt}
\end{assumpenum}
\end{assumption}

Assumption \ref{assump6_bdd_link} is met by many links considered in practice, such as the identity, logistic and probit links.
Assumption \ref{assump6_bdd_weights} is satisfied trivially by the identity link, and is true for the logistic and probit links as long as $f^*_P$ is bounded for each $P\in \mathcal{P}$.
Assumption \ref{assump6:bdd_alt} holds trivially with $M=0$ for all null distributions, and is a boundedness condition for local alternatives.

In this setting, we have $ s_P(x) = \mu(f^*_P(x))-\mu(g^*_P(x))$ where $\mu(g^*_P(X)):=\mathbb{E}_P(Y \given X)$.
In particular, the effect size of an alternative is given by \[ \tau_P = \mathbb{E}_P\{\mu(f^*_P(X))-\mu(g^*_P(X))\}^2. \]

\subsubsection{Conditional quantile specification} 

For conditional quantile specification, we have $\ell'(\eta,d) := \ind(y \leq \eta) - \tau$, where $\tau \in (0,1)$ is the quantile of interest.
Consider the following conditions:

\begin{assumption}\label{assump : quantile}
Suppose that the class of distributions $\mathcal{P}$ on $\mathbb{R}^d \times \mathbb{R}$ is such that $\operatorname{Var}_P(Y) < \infty$ for each $P \in \mathcal{P}$ and the following hold:
    \begin{assumpenum}
        \item The regular conditional distribution $\mathbb{P}_P(\cdot \given X=x)$ exists and admits a Lebesgue density given by $p_{Y|X}(\cdot \given x)$, which is uniformly bounded---there exists an $L > 0$ such that $p_{Y|X}(\cdot \given x) \leq L$ for each $x \in \mathbb{R}^d$ and $P \in \mathcal{P}$, and uniformly Lipschitz---there exists an $L' > 0$ such that $|p_{Y|X}(\eta \given x)-p_{Y|X}(\zeta \given x)| \leq L'|\eta -\zeta|$ for each $x \in \mathbb{R}^d$ and for each $P \in \mathcal{P}$.
        \label{assump7_density_reg}

       \item There exist $0 < c \leq C < 1$ such that $\mathbb{P}_P(Y \leq f^*_P(X) \given X = x) \in [c,C]$ for each $x \in \mathbb{R}^d$ and $\forall P \in \mathcal{P}.$ \label{assump7_minimiser_exist}

       \item There exists a $k > 0$ such that $p_{Y|X}(f^*_P(x) \given x) \geq k$ for each $x \in \mathbb{R}^d$ and $P \in \mathcal{P}$.
       \label{assump7_bdd_density}

    \end{assumpenum}
\end{assumption}

We verify in \cref{lem:quantile-assumptions} (see \cref{sec:lemmas}) that if \cref{assump : quantile} holds, then \cref{assump:risk_exist,assump:risk-regularity,assump:score_variance} are all satisfied by the class $\mathcal{P}$.
Here we have
$ s_P(x) = \mathbb{P}_P(Y \leq f^*_P(X) \given X = x) - \tau, $
 so that the effect size of an alternative is measured by 
\[ \tau_P = \mathbb{E}_P \left[\mathbb{P}_P(Y \leq f^*_P(X) \given X) - \tau \right]^2. \] 

\section{Numerical results} \label{sec:numerical}
In this section, we present numerical results on both synthetic and real data to illustrate the empirical performance of the Debiased Score Test (DST) for the various applications considered in \cref{subsec:applications}.
Unless stated otherwise, we use the procedure described in \cref{sec:hunting} for hunting, with the refinement of \cref{sec:accounting_DB} applied.
By default, we use \texttt{regression\_forest()} from the \texttt{grf} package \citep{grf} to hunt. 
The code for producing the numerical results is available at \url{https://github.com/dhawan-aditya/dst_simulations}.

\subsection{Generalised additive models} \label{sec:gam-simu}
We present some simulation studies for testing the conditional mean specification of a generalised additive model.
We generate covariates $X \sim \N_{10}(0,\Sigma)$, where $\Sigma_{ij} = 2^{-|i-j|}$.
Conditioned on $X$, we generate both continuous and binary responses:
\begin{description}
\item[Continuous] $Y = \sin(2X_1) + (\tau / \sqrt{n}) g(X) +v(X)\cdot \varepsilon $, where $\varepsilon \in \{\N(0,1), t_4\}$ is independent of $X$ and $v(X) = (1+X_2^2)^{1/2} / 2$;
\item[Binary] $\mathbb{P}(Y = 1 \given X) = \mu(\sin(2X_1) + (\tau / \sqrt{n}) g(X))$, where $\mu(\cdot)$ is either the logit or the probit link.
\end{description}
In both cases, we let $g(X) = X_1 X_3$ and use $\tau \geq 0$ to control the effect size: the null holds if and only if $\tau = 0$. 
To implement the DST described in \cref{sec:cond-mean}, we use package \texttt{grf} for hunting and use \texttt{bam()} from the R package \texttt{mgcv} \citep{wood2017} to fit all generalised additive models (with \texttt{method=`fREML'} for smoothing parameter selection and basis dimension $k = 20$). 
\cref{fig:gam-power} shows the resulting power curves under $n=2000$ ($n=1000$ and $5000$ produce similar results; see \cref{sec:gam-simu-extra}).
The power of the DST can be further boosted by aggregating the test statistics from 10 random splits using rank-transformed subsampling \citep{rt} (dashed curves).

We compare our test to an adaptation of the test in \citet{williamson2021general}, where we take the loss function to be the squared loss in the continuous case, and the negative log-likelihood of the binomial model in the binary case.
While their procedure was originally developed for assessing variable significance, the theory generalises to the goodness-of-fit setting as well, which we discuss in more detail in \cref{sec:williamson}. 
To implement the test, we fit the null model as above and the alternative model using \texttt{grf} with default settings.

As mentioned by the authors, the influence function for the target parameter in \citet{williamson2021general} becomes degenerate under the null (see \cref{sec:williamson}).
To maintain Type I error control, they introduce additional sample splitting to prevent the degeneracy.
However, this implies that their test only has power against a sequence of stronger-than-root-$n$ alternatives, namely when $\sqrt{n} \varepsilon_{n} \to \infty$ for $\varepsilon_{n}$ defined in \cref{sec:theory}.
Since this simulation setup has root-$n$ alternatives, i.e., $\sqrt{n} \varepsilon_{n}$ remains bounded, the adapted Williamson test is expected to have very low power and this is validated in \cref{fig:gam-power}.

The DST can also be applied to testing the specification of a partially linear model; we present the simulation results in \cref{sec:plm-simu}. 

\begin{figure}[htbp]
    \centering
    \includegraphics[width= 1\textwidth]{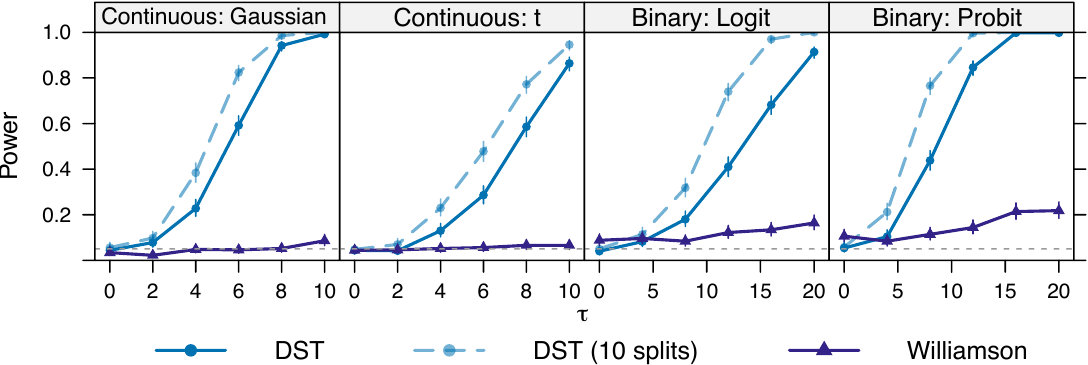}
    \caption{Power (with 95\% CI, dashed horizontal: $\alpha=0.05$) for testing the conditional mean specification of a GAM with $p=10$ covariates and $n=2000$. Dashed curves are from combining 10 splits of DST using the rank-transformed subsampling. }
    \label{fig:gam-power}
\end{figure}

\subsubsection{Model checking for insurance data} \label{sec:insurance}
We use our method to check model specification for the insurance dataset \texttt{FreMTPL2freq} from the R package \texttt{CASdatasets} \citep{CASdatasets2024}, which contains claim counts $Y$ and several risk features $X$ for motor third party liability claims in France.
\citet{kaggle} analysed this dataset, and considered various models for $Y$, the number of claims made during an exposure period of length $T \in (0,2]$ years.
Given covariates $X \in \mathcal{X}$, they model the claim count as
\begin{equation} \label{eq:poisson_model}
Y \given X, T \sim \pois( \exp(\Lambda(X))\cdot T),
\end{equation}
where $\Lambda$ is some function of the covariates $X$.
The authors considered various structural choices for $\Lambda$, and specifically an additive model
\begin{multline} \label{eq:additive_poisson_model}
\Lambda(X) = f_1(\texttt{VehAge}) + f_2(\texttt{DrivAge}) + f_3(\texttt{BonusMalus}) + \beta_1 \cdot \texttt{VehPower} + \beta_2 \cdot \texttt{VehGas} \\
+ \beta_3 \cdot \texttt{VehBrand} + \beta_4 \cdot \texttt{Area} + \beta_5 \cdot \texttt{Density} + \beta_6 \cdot \texttt{Region}. 
\end{multline}
The authors ultimately found that a LightGBM model \citep{ke2017lightgbm} outperforms the additive model, which suggests that the additivity of $\Lambda$ is likely misspecified.
Suppose we are confident that $\E(Y \given X, T) = \exp(\Lambda(X))T$, as implied by the Poisson model in \eqref{eq:poisson_model}, and would like to test the structural assumption that $\Lambda \in \mathcal{F}$, where $\mathcal{F}$ is the linear space of functions of the form \eqref{eq:additive_poisson_model}.
We take $\mathcal{G} = L^2(X)$ and use the negative Poisson log-likelihood as the loss function
\[ \ell(\eta,D):= -\eta Y+\exp(\eta)\cdot T, \quad D := (T, X, Y). \]
Under \eqref{eq:poisson_model}, this population risk is minimised over $\G$ at $\Lambda$.
Hence, we can use our method to test the specification \eqref{eq:additive_poisson_model} with score
$\ell'(\eta,D) = \exp(\eta)\, T-Y$, smoothed score $\phi(\eta,X) = \exp(\eta) \E(T \given X) -\E(Y \given X)$ and weight $w(X) = \exp(f(X)) \E(T \given X)$.

We compare the DST specification test with a test adapted from \citet{williamson2021general}, which is described in \cref{sec:williamson}. 
We use \texttt{bam} from the \texttt{mgcv} package to fit the null model for both methods.
The \texttt{grf} package is used for hunting in the DST, while the \texttt{lightgbm} package is used for fitting the alternative model for the \citet{williamson2021general} test.
We use 1/4 of the data as the test sample, from which we further draw 500 subsamples of varying sizes and compute a $p$-value on each.
The left panel of \cref{fig:insurance} plots the power curves against the size of the test subsample. 
To further assess calibration, we use the remaining 3/4 of the data to fit the null model $\hat{\Lambda}$, and then generate synthetic responses $\hat{Y}_i \sim \pois(\exp(\hat{\Lambda}(X_i))\cdot T_i)$ using $(T_i,X_i)$ from the test data.
The resulting Type I errors are shown in the right panel of \cref{fig:insurance}.

\begin{figure}[htbp]
    \centering
    \includegraphics[width=0.9\textwidth]{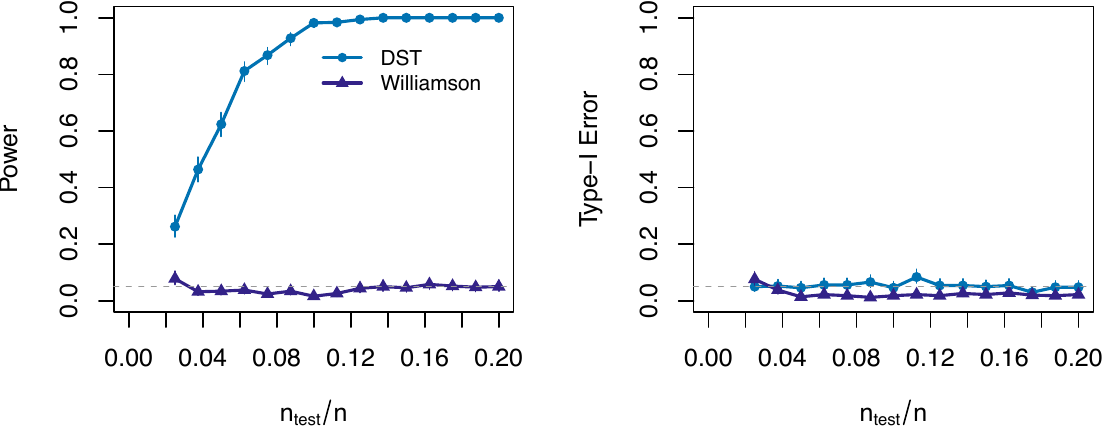}
    \caption{Power and Type I error (with 95\% CI) for testing the specification \eqref{eq:additive_poisson_model} on the insurance dataset. Left: power on the real test data; right: Type I error on the real test data but with synthetic null response (dashed horizontal line: $\alpha = 0.05$).}
    \label{fig:insurance}
\end{figure}

\subsection{Additive quantile regression} \label{sec:qgam-simu}
We operationalise the DST described in \cref{sec:quant} for testing the additivity of a quantile function:
we fit the additive quantile regression with the R package \texttt{qgam} \citep{qgam} using basis dimension $k = 20$, hunt with the \texttt{grf} package, and debias the hunted function with \texttt{bam()} from the \texttt{mgcv} package.
Additionally, recall that the weight $w(X) = p_{Y \given X}(f(X) \given X)$, i.e., the conditional density at the conditional quantile, is needed here.
We use a plug-in estimate: we first estimate the conditional density $p_{Y \given X}$ and then evaluate the estimate at $\hat{f}(X)$.
We run \texttt{grf}'s \texttt{quantile\_forest()} to derive the forest weights $(\alpha_i)_{1 \leq i \leq n}$, and then form the estimate $\hat{p}_{Y \given X}(y \given x) = h^{-1}\sum_{i=1}^{n} \alpha_i(x)\, K((y-Y_i)/h)$, where $K$ is the standard normal density and the bandwidth $h$ is selected using 5-fold cross-validation of the held-out log-likelihood.

We generate some simulations as follows.
Let $Y := f_1(X) + f_2(X) \cdot \varepsilon$, where $f_2 \geq 0$ and $\varepsilon \sim F$ is independent of the random vector $X$.
Then, for any $q \in (0,1)$, the $q$-th conditional quantile of $Y$ is given by $f_1(X) + f_2(X) \cdot F^{-1}(q)$, which is additive in the coordinates of $X$ if $f_1$ and $f_2$ are both additive. 
Specifically, we draw $X \sim \N_{10}(0,\Sigma)$ with $\Sigma_{ij} = 2^{-|i-j|}$ and then generate the response using 
\[f_1(X) = \sin(X_1) + \cos(X_2) +  X_3/4 + \tau/\sqrt{n} \cdot g(X), \; f_2(X) = (1 + X_2^2 + X_3^2)/4, \; g(X) = X_1 X_3,  \]
where $\tau \geq 0$ is the effect size that controls the departure from the null: the quantile function is additive if and only if $\tau = 0$. 
We focus on $q = 0.7$ and consider error distributions $F \in \{\N(0,1), t_3, \expo(1) \}$. 
\cref{fig:qgam-power} plots the power of the DST under sample size $n=2000$.
It additionally shows the curves for a variant of the DST that does not apply refinement for the hunt (see (c) in \cref{alg:dst_in}). We can see that the refinement slightly improves 
the power of the DST in the exponential example.

\begin{figure}[!htbp]
    \centering
    \includegraphics[width=0.9\textwidth]{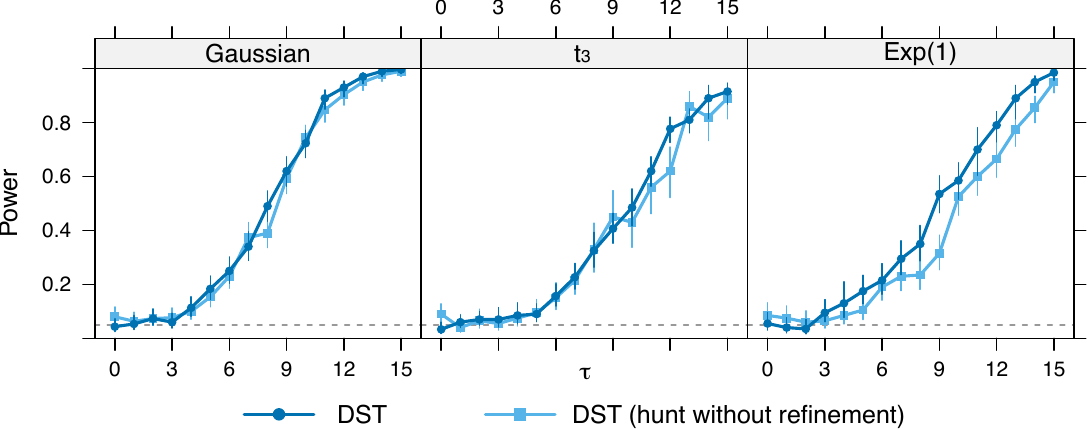}
    \caption{Power (with 95\% CI, dashed horizontal line: $\alpha=0.05$) of the DST for testing the additivity of a quantile function. }
    \label{fig:qgam-power}
\end{figure}

\subsection{Treatment effect heterogeneity}  \label{sec:hte-numerical}
In this section, we numerically study our method described in \cref{sec:hte} for detecting treatment effect heterogeneity.
We first present the results of simulation studies with the following data-generating mechanism with 6 covariates: 
\begin{align*}
& Z \sim \N_6(0, \Sigma), \quad \Sigma_{ij} = 5^{-|i-j|}, \quad T \given Z \sim \text{Ber}(\Phi(1/5 - Z_2 / 4 + Z_3 / 4)), \\
& \tau(z)  = 1 - z_3 / 2 + z_4^{2}, \quad \mu_0(z) = 2 \log\left(1 + \exp(z_1 + z_2 + z_3) \right), \\
& Y= \mu_0(Z) + (\gamma / \sqrt{n})\, T \cdot \tau(Z) + \varepsilon, \quad \varepsilon \sim \N(0,1/4).
\end{align*}
Here $\gamma \in \mathbb{R}$ is the effect size. 
We consider null hypotheses $H_0(S)$ that the CATE $\tau(z)$ only depends on $z$ through $z_{S^{c}}$ for the following choices of $S$: (i) all the covariates $S =\{1,\dots,6\}$, (ii) subset $S = \{3,4\}$ and (iii) singleton $S = \{3\}$. 
Rejecting $H_0(S)$ provides evidence for treatment effect heterogeneity in $Z_S$ while holding the other covariates fixed.
Choice (i) tests whether there is any heterogeneity at all.
For all these choices of $S$, $H_0(S)$ holds if and only if $\gamma = 0$. 

\cref{fig:hte-power} plots the power curves under sample size $n=2000$. 
We apply our method using \texttt{grf} to fit the nuisance functions.
As a method for comparison, we also test the null for $S = \{3,4\}$ and $S=\{3\}$ using the TE-VIM method of \citet{hines2025variable}, which assesses $H_0(S)$ via a variable importance measure (VIM) constructed by comparing the out-of-sample mean squared prediction errors of a CATE model fitted with and without $Z_S$.
Specifically, we compare to the variant `SS-B' of TE-VIM, which, like our method, employs data splitting and the R-learner for estimating the nuisances (also using \texttt{grf}).
As mentioned by the authors, TE-VIM has no formal guarantees, unless modified using the strategy of \citet{williamson2021general} --- this modified version is plotted in \cref{fig:hte-power}; for the unmodified version, see also \cref{fig:hte-tevim-unmod}.

For the case of $S = \{1,\dots,6\}$, we also compare to the built-in \texttt{test\_calibration()} of the \texttt{grf} package.
As discussed in \cref{subsec:grf}, it can be viewed as an omnibus test of $H_0(S)$ when $S$ is the full set and $T$ is binary.
This test, however, is not applicable when $S$ is a subset of covariates.
Due to the data splitting used in DST and TE-VIM, which produces randomised $p$-values with potentially reduced power, each of them is applied 10 times and the $p$-values are aggregated using the rank-transformed subsampling \citep{rt}; the \texttt{grf} test is
not improved by further aggregation.
The aggregated DST consistently achieves the highest power in all settings, while maintaining the Type I error at $\alpha = 0.05$ under the null. 
The modified TE-VIM has no power for detecting root-$n$ alternatives, but gains some power after aggregation.

\begin{figure}[htbp]
\centering
\includegraphics[width=.85\textwidth]{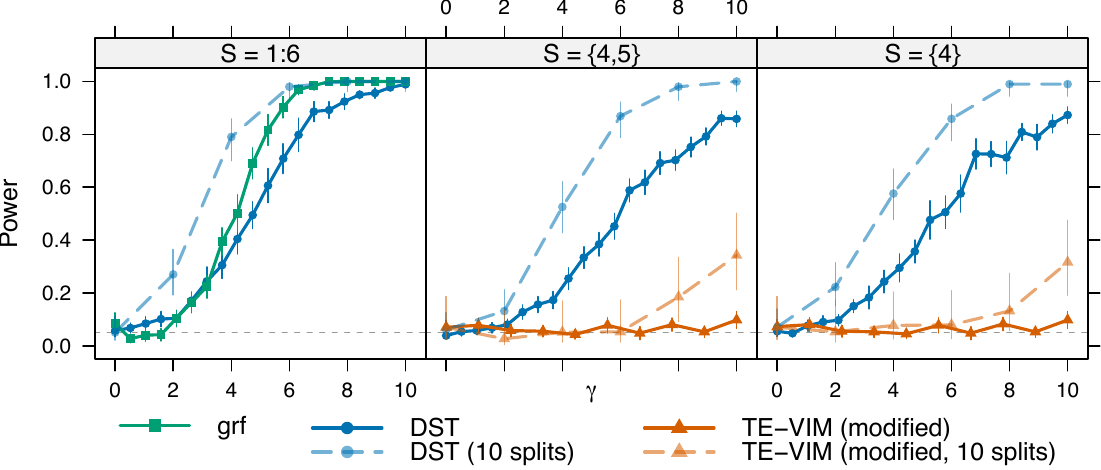}
\caption{Power (with 95\% CI) for detecting treatment effect heterogeneity in covariates $S$ (horizontal dashed: $\alpha = 0.05$). TE-VIM is modified with extra sample splitting to maintain calibration under the null. Dashed curves are based on aggregated tests from 10 data splits.}
\label{fig:hte-power}
\end{figure}

\subsubsection{Identifying effect modifiers for HIV treatment} \label{sec:hte-actg}
We also use our method to identify effect modifiers in treating HIV using the AIDS Clinical Trials Group (ACTG) Protocol 175 data \citep{hammer1996trial}, which is available from the R package \texttt{speff2trial}. 
This is a randomised trial that compares treatments for HIV among adults with CD4 T-cell counts between 200 and 500 cells per $\text{mm}^{3}$.
We consider two treatment arms: monotherapy with didanosine ($T=0$) versus combination therapy with zidovudine and didanosine ($T=1$), which have 561 and 522 patients respectively. 
We study the CD4 count at $20 \pm 5$ weeks as the continuous outcome $Y$, with higher values indicating better health. 
We are interested in the potential effect modification by 12 baseline covariates $Z$, consisting of 5 continuous variables (age, weight, Karnofsky score, baseline CD4 count, baseline CD8 count) and 7 binary variables (gender, homosexual activity, race, symptom, intravenous drug use, antiretroviral history, hamophilia).
For each covariate $Z_j$, we test the leave-one-out null hypothesis associated with $S = \{j\}$, which posits that the treatment effect does not vary by $Z_j$ while holding all the remaining 11 covariates fixed. 
Therefore, a small $p$-value indicates treatment effect heterogeneity in $Z_j$ not captured by the other covariates. 
Since the treatment is randomised, a constant propensity score is estimated.
Both methods employ cross fitting and use \texttt{grf} for fitting nuisance regressions.

\cref{fig:actg} reports the results. 
For each covariate, we apply DST and the modified TE-VIM over 100 random data splits and compute the aggregated $p$-value with the rank-transformed subsampling --- this is repeated 2000 times, yielding the boxplots reflecting the remaining variability after multiple splits. 
DST identifies the baseline CD4 count (median $p=0.03$) and homosexual activity (median $p=0.03$) as potentially significant effect modifiers, while the TE-VIM method yields no significance.

\begin{figure}[htbp]
\centering
\includegraphics[width=1.\textwidth]{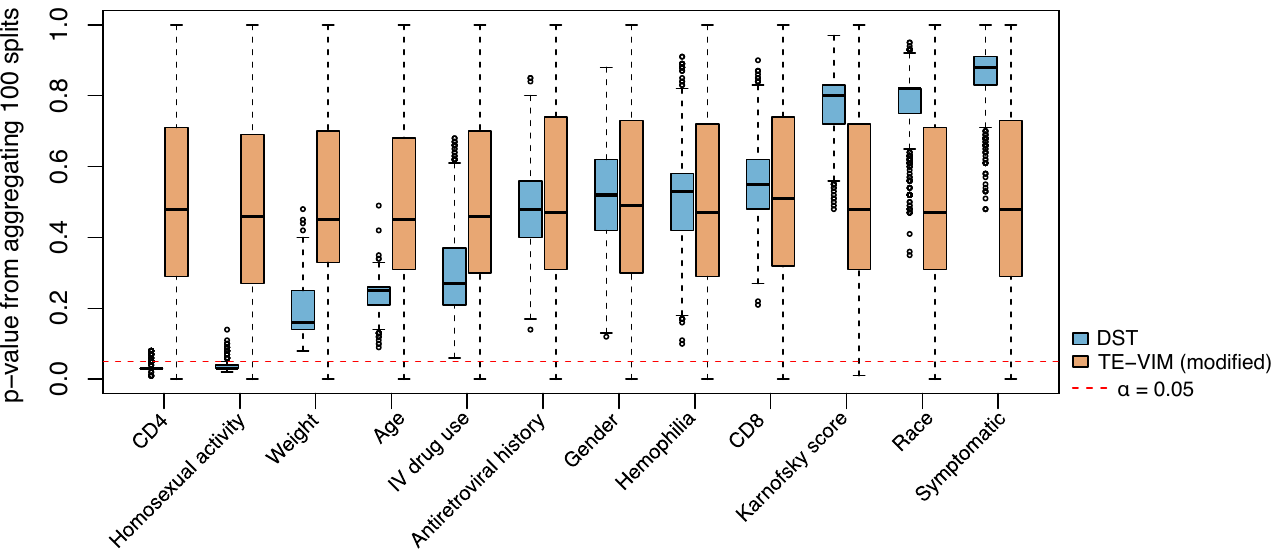}
\caption{Treatment effect heterogeneity in each baseline covariate for the AIDS clinical trial data, where the vertical axis shows the $p$-values aggregated from 100 random data splits, each calibrated using the rank-transformed subsampling. The boxplots reflect the remaining variability after aggregating over 100 splits.}
\label{fig:actg}
\end{figure}

\section{Discussion} \label{sec:discuss}
In this paper we introduced the Debiased Score Test (DST) for testing semiparametric hypotheses that leverages the predictive power of machine learning methods in a hunt-and-test strategy.
Building on this work, one interesting avenue for further research would be to extend the work to provide inference for the effect size $\tau_P$, such as a (one-sided) confidence interval.
Another potential line of investigation could consider what one should do after the DST rejects a hypothesis, and whether the hunted function could be used to refine the original model.
It would also be useful to extend some of the ideas here to settings with non-i.i.d.\ data, such as grouped data or data with temporal dependence.

\paragraph{Acknowledgements}
AD was supported in part by a grant from G-Research.
FRG was supported in part by NSF Grant DMS-2515385.
FRG would like to thank the Isaac Newton Institute for Mathematical Sciences, Cambridge, for support and hospitality during the programme \emph{Causal inference: From theory to practice and back again} (supported by EPSRC grant no.\ EP/K032208/1) where part of the work on this paper was undertaken.
The authors used Claude Code Fable 5 and Opus 4.8 for assistance with coding and proofreading the manuscript.

\bibliographystyle{plainnat}

\newpage

\newpage
\appendix
\appendixnumbering
\crefalias{section}{appendix}
\crefalias{subsection}{subappendix}
\crefalias{subsubsection}{subsubappendix}

This Supplementary Material consists of the following appendices.
\cref{sec:proofs_hunt} contains proofs and additional details pertaining to \cref{sec:dst}.
\cref{sec:GLM} describes the relationship between the DST and the classical score test for variable significance in generalised linear models.
\cref{sec:theory_proofs} contains the proofs and auxiliary lemmas for \cref{sec:theory} in the main text.
\cref{sec:simu-extra} describes implementation details of our method and relevant methods for comparison and also contains additional numerical results. 

\section{Proofs and additional details relating to \cref{sec:dst}} \label{sec:proofs_hunt}
\subsection{Regularity assumptions} \label{sec:regularity_hunt}
We make use of the following regularity assumptions in the results in \cref{sec:lemma_hunt}.
\begin{enumerate}
    \item $\mathcal{F}$ and $\mathcal{G}$ are closed linear subspaces of $L^2(X)$.
    \item There exists $0 < k \leq K$ such that $k \leq w(f^*(X), X) \leq K$ almost surely.
    \item There exists $0<c \leq C$ such that $c \leq \mathbb{E}(\ell'(f^*(X), D)^2 \given X) \leq C$ almost surely.
\end{enumerate}

\subsection{Lemmas} \label{sec:lemma_hunt}
In the following, we use the following notation.
For a positive random variable $U$, we will denote by $L^2(X,U)$ the set of functions $f : \mathcal{X} \to \mathbb{R}$ for which $\|f\|_U^2 := \mathbb{E}(Uf(X)^2) < \infty$. For $f \in L^2(X)$, we will write $\|f\|^2 := \E (f(X)^2)$.
Also, for a subspace $\mathcal{G} \subset L^2(X,U)$ and $h \in L^2(X,U)$, we will define
\[ \operatorname{proj}_{\mathcal{G}}^{U}(h) := \argmin_{g \in \mathcal{G}} \|h-g\|_U^2, \]
whenever a minimiser exists.
Note that such a minimiser is unique up to almost sure equivalence.

\begin{lemma} \label{lem:eqv_norm}
    Suppose $X \in \mathcal{X}$ and $U \in [0, \infty)$ are random variables such that $k \leq \mathbb{E}(U \given X) \leq K$ for some $0 <k \leq K$.
    Then, we have $L^2(X) = L^2(X,U)$ and the two spaces have equivalent norms.
    In particular, $\mathcal{G} \subset \mathbb{R}^\mathcal{X}$ is closed in $L^2(X)$ if and only if it is closed in $L^2(X,U)$.
\end{lemma}

\begin{proof}
    For any measurable $h : \mathcal{X} \to \mathbb{R}$, recall $\|h\|_U^2 := \mathbb{E}(Uh(X)^2)$ and $\|h\|^2 = \mathbb{E}(h(X)^2)$.
    Then, note that
\[ k\|h\|^2 \leq \mathbb{E}(\mathbb{E}(U \given X)h(X)^2) = \mathbb{E}(Uh(X)^2) = \|h\|_U^2\leq K\|h\|^2, \]
    and so in particular the identity map $\operatorname{id} : L^2(X) \to L^2(X,U)$ is well defined and a homeomorphism.
\end{proof}

\begin{lemma}[Orthogonality condition]\label{lem:wls-orth}
    Let $0 \leq w(X) \leq K$ for some $K \geq 0$ and $g \in L^2(X)$.
    Suppose $\mathcal{F}$ is a linear subspace of $L^2(X)$, and let $m_g \in \mathcal{F}$ be a minimiser of the weighted least squares objective
$$ \mathbb{E}[w(X)\{g(X) - f(X)\}^2] $$
    over $f \in \mathcal{F}$.
    Then, $m_g$ satisfies the orthogonality condition
$$ \mathbb{E}[\{g(X)-m_g(X)\}w(X)f(X) ] = 0 \ \forall f \in \mathcal{F}. $$
\end{lemma}
\begin{proof}
    For $f \in \mathcal{F}$, define $\psi(t) := \mathbb{E}[\{g(X) - m_g(X) - tf(X)\}^2w(X)]$ for $t \in [-1,1]$.
    Then $\psi$ is differentiable and attains its minimum at $t = 0$, and so we must have $\psi'(0) = -2\mathbb{E}[\{g(X)-m_g(X)\}w(X)f(X)] = 0$.
\end{proof}

\begin{lemma}[Existence of projection]\label{lem:proj-existence}
    Let $k \leq w(X) \leq K$ for some $0 < k \leq K$.
    Suppose $\mathcal{F}$ is a closed linear subspace of $L^2(X)$, and $g \in L^2(X)$.
    Then the minimiser of
\[ \mathbb{E}[w(X)\{g(X) - f(X)\}^2] \]
    over $f \in \mathcal{F}$ exists and is unique.
\end{lemma}

\begin{proof}
Firstly, note by \cref{lem:eqv_norm} that $L^2(X,w(X)) \cong L^2(X)$ is a Hilbert space, so that projections onto closed subspaces are well defined.
Moreover, note that the minimiser of the above objective is precisely the projection of $g$ onto $\mathcal{F}$ in $L^2(X,w(X))$ whenever the latter is well defined.
Finally, since $\mathcal{F}$ is closed in $L^2(X)$, it is also closed in $L^2(X,w(X))$ by \cref{lem:eqv_norm} and thus the projection is indeed well defined.
\end{proof}

\subsection{Proofs}
\subsubsection{Proof of \cref{prop:orthog}}
\begin{proof}
    Directly follows from \cref{lem:wls-orth} and \cref{lem:proj-existence}.
\end{proof}

\subsubsection{Proof of \cref{prop:hunt}}
\begin{proof}
 Let us write $R:=\ell'(f^*(X), D)$ for notational ease.
 Suppose first that $\hat{h}$ minimises \eqref{eq:recip_wls} over $h \in \mathcal{G}$.
 Note that then $m := \E [(R \hat{h}(X))^2] < \infty$ and so writing $\hat{H} = \hat{h}(X)$, we see that $\hat{H}$ minimises
\begin{equation} \label{eq:H_optim}
\Xi(H) := \E[R^2 (R^{-1} - H)^2] = 1 - 2 \E (RH) + \E(R^2H^2)
\end{equation}
    over $H \in \mathcal{H} := \{h(X) : h \in \mathcal{G} \text{ and } R h(X) \in L^2(X)\}$.
    Observe that then $\hat{H}$ maximises $\E (RH)$ over $H \in \mathcal{H}$ subject to $\E(R^2H^2) = m$.

    Now $\Lambda(ch) = \Lambda(h)$ for any $c>0$, so in order to maximise $\Lambda$, it suffices to maximise it over $h$ such that $\E (R^2h(X)^2) = m$.
    In other words, we may maximise
\[ \frac{\E(RH)}{\sqrt{m - [\E(RH)]^2}} \]
    over $H \in \mathcal{H}$ with $\E (R^2 H^2) = m$.
    The above is monotone increasing in $\E(RH)$ when this is non-negative, so it suffices to maximise $\E(RH)$ over $H \in \mathcal{H}$ with $\E (R^2 H^2) = m$, and hence $\hat{H}$ is a maximiser.

    To prove the final part of the statement, suppose now that $\hat{h}$ maximises $\Lambda(h)$ over $h \in \mathcal{G}$.
    Given $H \in \mathcal{H}$, let $m := \E(R^2 H^2)$.
    Let $\tilde{H} := m^{1/2} \hat{h}(X) / \sqrt{\E(R^2 \hat{h}(X)^2)}$, so $\E(R^2\tilde{H}^2) = m$.
    Furthermore, let $\hat{H} = \E(R \hat{h}(X)) \hat{h}(X) / \sqrt{\E(R^2 \hat{h}(X)^2)}$, so $\Xi(\hat{H}) = \inf_{c >0} \Xi(c\, \hat{h}(X))$.
    Then
\[ \Xi(H) \geq \Xi(\tilde{H}) \geq \Xi(\hat{H}), \]
    but $H$ was arbitrary so $\hat{H}$ is a minimiser of $\Xi$ over $\mathcal{H}$.
\end{proof}

\subsubsection{Proof of \cref{prop:opt_hunt}}
\begin{proof}
 Let us write $R:=\ell'(f^*(X), D)$ for notational ease. 
 Note that the spaces
  \[
  \mathcal{G}:=L^2(X), \qquad L^2(X,R^2), \qquad L^2(X,v(X)), \qquad
  \text{and} \qquad L^2(X,w(X))
  \]
  all coincide and have equivalent norms by \cref{lem:eqv_norm}, as $c \leq v(X) := \mathbb{E}(R^2 \given X) \leq C$ and $k \leq w(X) \leq K$. By definition, $h^* = \operatorname{proj}^{R^2}_{\mathcal{G}}(R^{-1})$ and $h^*_{\mathrm{DB}} = \operatorname{proj}^{R^2}_{\mathcal{G}_{\mathrm{DB}}}(R^{-1})$ whenever the projections exist, and are thus unique.
 Also, $\mathcal{G}_{\DB}$ is the orthogonal complement of $\mathcal{F}$ in $L^2(X,w(X))$ and is thus closed in $L^2(X,w(X))$, and hence in $L^2(X,R^2)$. Thus, the projections above are indeed well defined by Hilbert's Projection Theorem, and thus $h^*$ and $h^*_{\mathrm{DB}}$ are well defined.

Now, since $\mathcal{G}_{\mathrm{DB}} \subset \mathcal{G}$, we have that $h^*_{ \mathrm{DB}} = \operatorname{proj}^{R^2}_{\mathcal{G}_{\mathrm{DB}}}(h^*)$.
Unravelling the definition of the projection, we see thus
\begin{align*}
h^*_{\mathrm{DB}} &= \argmin_{g \in \mathcal{G}_{\mathrm{DB}}} \mathbb{E}(R^2 \{h^*(X)-g(X)\}^2) = \argmin_{g \in \mathcal{G}_{\mathrm{DB}}} \mathbb{E}(v(X) \{h^*(X)-g(X)\}^2),
\end{align*}
so that $h^*_{\mathrm{DB}} = \operatorname{proj}^{v}_{\mathcal{G}_\mathrm{DB}}(h^*)$.
Again, by \cref{lem:eqv_norm}, $\mathcal{G}_{\mathrm{DB}}$ is closed in $L^2(X,v(X))$, and so we have that $h^* - h^*_{\mathrm{DB}} = \operatorname{proj}^{v}_{\mathcal{G}_{\mathrm{DB}}^{\perp}}(h^*)$, where $\mathcal{G}_{\mathrm{DB}}^{\perp}$ is the orthogonal complement of $\mathcal{G}_{\mathrm{DB}}$ in $L^2(X,v(X))$.

Now, note that $h \in \mathcal{G}_{\mathrm{DB}}$ if and only if
\[ \mathbb{E}(h(X)f(X)w(X)) = 0 \ \forall f \in \mathcal{F}, \]
which is the same as asking
\[ \mathbb{E}(v(X)h(X)g(X)) = 0 \ \forall g \in w\mathcal{F}/v. \]
Hence, we observe that $\mathcal{G}_{\mathrm{DB}}$ is the orthogonal complement of $w\mathcal{F}/v$ in $L^2(X,v(X))$. Moreover, $w\mathcal{F}/v$ is closed in $L^2(X,v(X))$. Indeed, the map $f\mapsto wf/v$ is a homeomorphism from
  $\mathcal{F}$ onto its image. Since $\mathcal{F}$ is closed,
  $w\mathcal{F}/v$ is closed in $L^2(X,v(X))$,
and thus $\mathcal{G}_{\mathrm{DB}}^{\perp} = w\mathcal{F}/v.$ Then, writing $h^*-h^*_{\mathrm{DB}} = wf^*_{\mathrm{DB}}/v$, we have that
\[ f^*_{\mathrm{DB}} = \argmin_{f \in \mathcal{F}} \mathbb{E}\!\left[ \frac{w(X)^2}{v(X)} \left( \frac{v(X)h^*(X)}{w(X)} - f(X) \right)^2 \right]. \]
\end{proof}

\section{Comparison with the classical score test for generalised linear models} \label{sec:GLM}
It is instructive to compare the DST to the classical score test in the case of a generalised linear model with canonical link. Suppose $X = (T, Z) \in \R \times \R^p$ and we wish to test the significance of the predictor of interest $T$. In this case, the score test based on the data $I_1$ would take the form of a studentised version of
\begin{equation} \label{eq:classical_score1}
\frac{1}{\sqrt{n}} \sum_{i \in I_1} T_i \{Y_i - \mu(\tilde{c} + Z_i^{\top} \tilde{\beta} )\}.
\end{equation}
Here $(\tilde{c}, \tilde{\beta})$ correspond to maximum likelihood estimates in the restricted model where the coefficient corresponding to $T$ is set to zero. These satisfy the score equations
\[
\sum_{i \in I_1} Z_i \{Y_i - \mu(\tilde{c} + Z_i^{\top} \tilde{\beta} )\} = 0, \qquad \sum_{i \in I_1} \{Y_i - \mu(\tilde{c} + Z_i^{\top} \tilde{\beta} )\} = 0,
\]
so in particular, \eqref{eq:classical_score1} is equivalent to
\begin{equation} \label{eq:classical_score2}
\frac{1}{\sqrt{n}} \sum_{i \in I_1} \{T_i - k - q^{\top} Z_i\}\{Y_i - \mu(\tilde{c} + Z_i^{\top} \tilde{\beta} )\}
\end{equation}
for any $(k, q) \in \R \times \R^p$.

Now consider applying the DST with the particularly simple choice of $\hat{h}$ given by $\hat{h}(X_i) = \tilde{s} T_i$ where $\tilde{s} = \sgn(\sum_{i \in I_1} \{Y_i - \mu(\tilde{c} + Z_i^{\top} \tilde{\beta} )\} )$\footnote{Note that our framework places no explicit restriction on the form of $\hat{h}$, though as discussed earlier, typically to gain power against diverse alternatives, we recommend the use of flexible regression methods in constructing it.}.
This toy version of the DST yields a test statistic of the form given by \eqref{eq:classical_score2}, though multiplied by $\tilde{s}$ and with the sum over $i \in I_2$ (the debiasing function $\hat{m}_{\hat{h}} \in \mathcal{F}$ is necessarily of the form $\hat{m}_{\hat{h}}(X_i) = k + q^{\top}Z_i$ for some $k$ and $q$). Another difference is that the DST uses two samples and performs a one-sided rather than a two-sided test, with $\tilde{s}$ carrying information from the first sample about the direction in which to test.
Although avoiding sample splitting makes the classical score test more efficient in this very simple setting, the independence created by sample splitting permits $\hat{h}$ to be chosen essentially arbitrarily.
In settings with diverse alternatives, this flexibility allows $\hat{h}$ to be constructed using powerful regression or machine learning methods, yielding tests with good power against a much broader range of departures from the null.

\section{Proofs of results in \cref{sec:theory}} \label{sec:theory_proofs}
In this section, we present proofs of all our results in \cref{sec:theory} and related auxiliary lemmas.

\subsection{Lemmas} \label{sec:lemmas}
\begin{lemma}[Loss Lipschitz bound] \label{lem:loss_reg}
    Consider the setup in \cref{sec:theory}, and suppose Assumption~\ref{assump4_bdd-weights} holds.
    Then, for each $\eta,\zeta \in \mathbb{R}$, we have that
\[ \sup_{P \in \mathcal{P}}\mathbb{E}_P[|\ell'(\eta,D)-\ell'(\zeta,D)| \given X] \leq K|\eta -\zeta|. \]
\end{lemma}

\begin{proof}
Since $\eta \mapsto \ell(\eta,d)$ is convex for each $d$, the path of subgradients $\eta \mapsto \ell'(\eta,d)$ is non-decreasing.
Hence, for any $\eta,\zeta \in \mathbb{R}$,
\begin{align*}
|\ell'(\eta,D)-\ell'(\zeta,D)| &= \{\ell'(\eta,D)-\ell'(\zeta,D)\}\,\ind(\eta \ge \zeta) \\
&\quad + \{\ell'(\zeta,D)-\ell'(\eta,D)\}\,\ind(\zeta > \eta).
\end{align*}

Taking conditional expectations given $X$, we obtain
\[ \mathbb{E}_P\!\left[ |\ell'(\eta,D)-\ell'(\zeta,D)| \,\big|\, X \right] = |\phi_P(\eta,X)-\phi_P(\zeta,X)|. \]

By Assumption~\ref{assump4_bdd-weights}, and the mean value theorem,
\[ |\phi_P(\eta,X)-\phi_P(\zeta,X)| \le \sup_{\theta\in\mathbb{R}} |\phi'_P(\theta,X)|\,|\eta-\zeta| \le K\,|\eta-\zeta|, \]
which concludes the proof.
\end{proof}

\begin{lemma}[Derivatives of risk]\label{lem:riskfunc}
    Consider the setup in \cref{sec:theory}, where $P \in \mathcal{P}$, and define $\forall g,h \in \mathcal{G}_P$ the function $\gamma_P : (-1,1) \to \mathbb{R}$ by $\gamma_P(t) := R_P (g +th)$.
    Then, $\gamma_P$ is twice differentiable at $0$, with $\gamma_P'(0) = \mathbb{E}_P[\ell'(g(X),D)h(X)]$ and $\gamma_P''(0) = \mathbb{E}_P[\phi'_P(g(X), X)h(X)^2]$.
\end{lemma}

\begin{proof}
Recall that $R_P(g) = \mathbb{E}_P(\ell(g(X),D))$, so that we have
\begin{align*}
\gamma_P(t) &= \mathbb{E}_P[\ell(g(X) +th(X),D)] \\
&= \mathbb{E}_P\left[\ell(g(X),D)+\int_0^t\ell'(g(X) + u h(X),D)h(X)du\right] {\text{(Fundamental Theorem of Calculus)}}{}\\
&= \gamma_P(0) +t\mathbb{E}_P[\ell'(g(X),D)h(X)] +{\mathbb{E}_P\left[\int_0^t\{\ell'(g(X) + u h(X),D) - \ell'(g(X),D)\}h(X)du \right] }
\end{align*}
 Now note that
  \begin{align*}
  &\left|
  \mathbb{E}_P\!\left[
  \int_0^t
  \{\ell'(g(X)+uh(X),D)-\ell'(g(X),D)\}h(X)\,du
  \right]
  \right|\\
  &\qquad\le
  \int_{\min(0,t)}^{\max(0,t)}
  \mathbb{E}_P\!\left[
  |\ell'(g(X)+uh(X),D)-\ell'(g(X),D)|\,|h(X)|
  \right]du.
  \end{align*}
  By \cref{lem:loss_reg}, for every \(u\in(-1,1)\),
  \[
  \mathbb{E}_P\!\left[
  |\ell'(g(X)+uh(X),D)-\ell'(g(X),D)|
  \,\middle|\,X
  \right]
  \le K|u|\,|h(X)|.
  \]
  Therefore,
  \begin{align*}
  &\left|
  \mathbb{E}_P\!\left[
  \int_0^t
  \{\ell'(g(X)+uh(X),D)-\ell'(g(X),D)\}h(X)\,du
  \right]
  \right|\\
  &\qquad\le
  \int_{\min(0,t)}^{\max(0,t)}
  K|u|\,\mathbb{E}_P[h(X)^2]\,du\\
  &\qquad=
  \frac{Kt^2}{2}\,\mathbb{E}_P[h(X)^2].
  \end{align*}

Thus, we have that
\begin{equation*}
\gamma_P(t) = \gamma_P(0) + t\mathbb{E}_P[\ell'(g(X),D)h(X)] + o(t),
\end{equation*}
and it follows that $\gamma_P'(0) = \mathbb{E}_P[\ell'(g(X),D)h(X)]$.
The same argument shows that
\[ \gamma'_P(t) = \mathbb{E}_P[\ell'(g(X)+th(X),D)h(X)] = \mathbb{E}_P[\phi_P(g(X)+th(X),X)h(X)]. \]
Note that $|\phi_P'(g(X) + th(X),X)h(X)^2| \leq Kh(X)^2$, and so differentiating under the integral sign gives
\[ \gamma''_P(t) = \mathbb{E}_P[\phi'_P(g(X)+th(X), X)h(X)^2]. \]
It follows that
\begin{align*}
\gamma'_P(0) &= \mathbb{E}_P[\ell'(g(X),D)h(X)] \\
\gamma''_P(0)&= \mathbb{E}_P[\phi'_P(g(X), X)h(X)^2]. \qedhere
\end{align*}
\end{proof}

\begin{lemma}[First order condition for risk minimisers] \label{lem:risklocmin}
Consider the setup in \cref{sec:theory}, and suppose $\mathcal{G}_P$ is a subspace of $L^2_P(X)$.
Then
\[
g^*_P \in \argmin_{g \in \mathcal{G}_P} R_P(g) \Longleftrightarrow \mathbb{E}_P[\ell'(g^*_P(X),D)h(X)] = 0 \qquad \forall h \in \mathcal{G}_P.
\]
\end{lemma}

\begin{proof}
    
Note that $g^*_P$ minimises the risk over $\mathcal{G}_P$ if and only if it minimises the risk over all segments of the form $\{g^*_P +th, t \in (-1,1) \}$, where $h \in \mathcal{G}_P$.
To this end, define $\gamma_P : (-1,1) \to \mathbb{R}$ by
\[ \gamma_P(t) = R_P(g^*_P + th). \]

As we assume the loss function $\ell(\cdot,d)$ to be convex for each $d\in \mathcal{D}$, it follows that $\gamma_P$ is also convex.
It is thus enough to show that $\gamma_P$ is minimised at $0$, but by convexity this is the case if and only if $\gamma_P'(0) = 0$, and by \cref{lem:riskfunc} we have $\gamma_P'(0) = \mathbb{E}_P[\ell'(g^*_P(X),D)h(X)]$.
\end{proof}
 
\begin{lemma}[Bound for remainder term] \label{lem:rem_term}
    Consider the setup in \cref{sec:theory}.
    Then, for fixed $x \in \mathcal{X}$ and $y_1, y_2 \in \mathbb{R}$, $|\phi_P(y_2,x) - \phi_P(y_1,x) -\phi'_P(y_1,x)(y_2-y_1)| \leq L'(y_2 -y_1)^2$.
\end{lemma}

\begin{proof}
By the mean value theorem, $\phi_P(y_2,x) - \phi_P(y_1,x) = \phi'_P(y,x)(y_2-y_1)$ for some $y$ between $y_1$ and $y_2$.
Hence,
\begin{align*}
|\phi_P(y_2,x) - \phi_P(y_1,x) -\phi'_P(y_1,x)(y_2-y_1)| &= |\phi'_P(y,x)-\phi'_P(y_1,x)|\,|y_1-y_2| \\
&\leq L'|y-y_1|\,|y_2 - y_1| \ (\text{Using Assumption \ref{assump4_lipschitz-weights}})\\
&\leq L'(y_2 - y_1)^2. \qedhere
\end{align*}
\end{proof}

\begin{lemma}[Risk for conditional mean specification]\label{lem:loss-gam}
    Let $\mu : \mathbb{R} \to \mathbb{R}$ be an $L$-Lipschitz strictly increasing function, and let $K$ be an antiderivative of $\mu$.
    Define the loss function $\ell : \mathbb{R} \times \mathcal{D} \to \mathbb{R}$ by
\begin{equation*}
\ell(\eta,d) = -\eta y+K(\eta),
\end{equation*}
    where $\mathcal{D} = \mathbb{R}^d \times\mathbb{R}$.
    Suppose $D= (X,Y) \in \mathcal{D}$ has distribution $P$, with $\var_P(Y) < \infty$.
    Suppose $\mathbb{E}_P(Y \given X) = \mu(g^*(X))$ for some $g^* \in L^2_P(X)$.
    Then, the risk
\[ R_P(g) := \mathbb{E}_P(\ell(g(X),D)) \]
    is well-defined $\forall g \in L^2_P(X)$, convex and
\[ g^*\in \argmin_{g \in L^2_P(X)} R_P(g). \]
\end{lemma}

\begin{proof}
Note that by the mean value theorem, \[K(\theta) = K(0) + \theta\mu(\theta_0),\] for some $\theta_0$ between $0$ and $\theta$.
Hence,
\begin{align*}
|K(\theta)| &\leq |K(0)| + |\theta \mu(\theta_0)| \\
&\leq |K(0)| + |\theta\mu(0)| + |\theta| \cdot \underbrace{|\mu(\theta) - \mu(\theta_0)|}_{\leq L|\theta-\theta_0| \leq L|\theta|} \\
&\leq |K(0)| + |\theta\mu(0)|+L\theta^2.
\end{align*}
Thus $\mathbb{E}_P|K(g(X))| < \infty$ when $g \in L^2_P(X)$.
Moreover,
\begin{equation*}
\mathbb{E}_P|Yg(X)| < \infty,
\end{equation*}
as $\mathbb{E}_P(Y^2)$ and $\mathbb{E}_P(g(X)^2) < \infty$.
It follows that $R_P(g)$ is well-defined for all $g \in L^2_P(X)$.
Note that $K' = \mu$, which is increasing, so $K$ is convex.
Thus, $\eta \mapsto \ell(\eta,d)$ is convex, and consequently the risk is convex.

Note that
\[ \mathbb{E}_P(\ell(\eta,D) \given X) = -\eta\cdot \mathbb{E}_P(Y \given X) + K(\eta) \]
is convex in $\eta$, and so minimal when $\mu(\eta) =\mathbb{E}_P(Y|X)$, which is the case precisely when $\eta = g^*(X)$.
Thus, $\forall g \in L^2_P(X)$, we have that
\[ \mathbb{E}_P(\ell(g^*(X),D) \given X) \leq \mathbb{E}_P(\ell(g(X),D) \given X) , \]
so that taking expectations gives $R_P(g^*) \leq R_P(g)$, showing that \[ g^*\in \argmin_{g \in L^2_P(X)} R_P(g).
\qedhere \]
\end{proof}

\begin{lemma}[Risk for conditional quantile specification] \label{lem:conditional_quantile}
Suppose $\tau \in (0,1)$, and define the loss function $\ell : \mathbb{R} \times \mathcal{D} \to \mathbb{R}$ by
\[ \ell(\eta,d) = \rho_{\tau}(y - \eta), \]
where $\rho_\tau$ is as defined in \cref{sec:quant} and $\mathcal{D}= \mathbb{R}^d \times \mathbb{R}$.
Suppose $D = (X,Y) \in \mathcal{D}$ has distribution $P$, with $\var_P(Y) < \infty$.
Suppose that the regular conditional distribution of $Y$ given $X$ exists, and define $q_\tau(x) := \inf \left\{ y \in \mathbb{R} : \mathbb{P}(Y \leq y \given X = x) \geq \tau \right\}$.
Then, the risk
\[ R_P(g) := \mathbb{E}_P(\ell(g(X),D)) \]
is well defined for all $g \in L^2_P(X)$, is convex and
\[ q_\tau \in \argmin_{g \in L^2_P(X)} R_P(g). \]
\end{lemma}

\begin{proof}
    To begin, note that $x \mapsto \rho_\tau(y-g(x))$ is measurable for each $y \in \mathbb{R}$ whenever $g$ is, and that $|\ell(g(X),D)| = |\rho_\tau(Y - g(X))| \leq |Y| + |g(X)|$, and so $R_P(g) = \mathbb{E}_P(\ell(g(X),D))$ is well-defined whenever $g \in L^1_P(X)$, and so certainly whenever $g \in L^2_P(X)$.
    Moreover, since $\eta \mapsto \rho_\tau(\eta)$ is convex, so is $\eta \mapsto \ell(\eta,d)$ for each $d \in \mathcal{D}$, and thence so is the risk.

Note that since $x \mapsto \mathbb{P}_P(Y \leq y \given X = x)$ is measurable for each $y$, so is the infimum (taken over the rationals without loss of generality) $x \mapsto q_\tau(x)$. Now when $q_{\tau}(x) \geq 0$, we have that $\mathbb{E}_P(Y^2 \given X = x) \geq (1-\tau)q^2_{\tau}(x)$. When $q_{\tau}(x) < 0$, we have that $\mathbb{E}_P(Y^2 \given X) \geq \tau q^2_{\tau}(x)$. Thus, we have that $q^2_\tau(X) \leq c \cdot \mathbb{E}(Y^2 \given X)$, where $c = \max(\frac{1}{\tau},\frac{1}{1-\tau})$, so that taking expectations
$\mathbb{E}_P(q_\tau^2(X)) \leq c \cdot\mathbb{E}_P(Y^2) < \infty$, and so $q_\tau \in L^2_P(X).$ It is well known that the $\tau$-quantile of a distribution $Q$ minimises $\eta \mapsto \mathbb{E}_Q(\rho_\tau(y - \eta))$, and thus we have that $\mathbb{E}_P(\ell(q_\tau(X),D) \given X) \leq \mathbb{E}_P(\ell(g(X),D) \given X)$ for each $g \in L^2_P(X)$.
    Taking expectations, we have that $R_P(q_\tau) \leq R_P(g) \ \forall g \in L^2_P(X)$.
\end{proof}

\begin{lemma} \label{lem:alt_proj}
    Consider the setup in \cref{sec:theory}.
    Then, the minimiser
\[ \min_{g \in \mathcal{G}_P }\E_P \{g(X)-\ell'(f^*_P(X),D)\}^2, \]
exists.
Moreover, it is unique up to almost sure equivalence under the law of $X$ under $P$.
\end{lemma}

\begin{proof}
    Consider the natural embedding of $\mathcal{G}_P$ in the Hilbert space $L^2_P$ of square integrable random variables via $g \mapsto g(X)$.
    Then, this embedding of $\mathcal{G}_P$ remains closed in $L^2_P$, and therefore the minimiser we seek is precisely the projection of $\ell'(f^*_P(X),D)$ onto $\mathcal{G}_P$, which exists by the Hilbert space projection theorem.
\end{proof}

\begin{lemma} \label{lem:snr_bound}
    Consider the setup in \cref{sec:theory}, with $P \in \mathcal{P}$. For $g \in \mathcal{G}_{\mathrm{DB}}$ satisfying
    \[
    \operatorname{Corr}_P(g(X),s_P(X))>0,
    \]
    we have the bound
\[
\Lambda_P(g) \geq A{\tau_P}^{1/2}\operatorname{Corr}_P(g(X),s_P(X)),
\]
where
\[
A=\frac{1}{C^{1/2}}\left(1+\frac{(K-k)^2}{4k^2}\right)^{-1/2}>0.
\]
\end{lemma}

\begin{proof}
    Note that
    \begin{align*}
    \var_P(\ell'(f^*_P(X),D)g(X))
    &\leq \mathbb{E}_P[\ell'(f^*_P(X),D)^2g(X)^2] \\
    &= \mathbb{E}_P[\mathbb{E}_P(\ell'(f^*_P(X),D)^2 \given X)g(X)^2] \\
    &\leq C\mathbb{E}_P[g(X)^2].
    \end{align*}

Now since $g \in \mathcal{G}_{\mathrm{DB}}$, and $\mathcal{F}_P$ contains the constant functions, we have that $\mathbb{E}_P(w(X)g(X)) = 0$. Thus by the Cauchy--Schwarz inequality, we have that

\begin{align*}
   \mathbb{E}_P(w(X))^2\cdot\mathbb{E}_P(g(X))^2
   &= \operatorname{Cov}_P(g(X),w(X))^2 \\
   &\leq \operatorname{Var}_P(g(X))\operatorname{Var}_P(w(X)) \\
   &\leq \frac{(K-k)^2}{4}\operatorname{Var}_P(g(X)),
\end{align*}

so that
\[
\mathbb{E}_P(g(X))^2
\leq
\frac{(K-k)^2}{4k^2}\operatorname{Var}_P(g(X)).
\]
Hence,
\[
\mathbb{E}_P[g(X)^2]
=
\operatorname{Var}_P(g(X))
+
\mathbb{E}_P(g(X))^2
\leq
\left(
1+\frac{(K-k)^2}{4k^2}
\right)\operatorname{Var}_P(g(X)).
\]

Thus,
\begin{align*}
\Lambda_P(g)
&= \frac{\mathbb{E}_P(\ell'(f^*_P(X),D)g(X))}{\sqrt{\var(\ell'(f^*_P(X),D)g(X))}} \\
&= \frac{\mathbb{E}_P(s_P(X)g(X))}{\sqrt{\var(\ell'(f^*_P(X),D)g(X))}}\\
&\geq A \cdot \frac{\mathbb{E}_P(s_P(X)g(X))}{\sqrt{\operatorname{Var}_P(g(X))}} \\
&= A \tau_P^{1/2}\operatorname{Corr}_P(s_P(X),g(X)),
\end{align*}
where the inequality follows since
\(\operatorname{Corr}_P(g(X),s_P(X))>0\), and the final equality follows by noting that $\mathbb{E}_P(s_P(X)) = 0$ since $\mathcal{F}_P$ contains the constant functions, and thereafter \cref{lem:risklocmin}.

\end{proof}
\begin{lemma} \label{lem:conditional_mean}
    Consider the conditional mean specification setup in \cref{subsubsec:cmean-theory}, and suppose \cref{assump:mean} holds.
    Then, the class of distributions $\mathcal{P}$ satisfies \cref{assump:risk_exist,assump:risk-regularity,assump:score_variance} in \cref{sec:theory}.
\end{lemma}

\begin{proof}
    Following the proof of \cref{lem:loss-gam}, we see that \cref{assump:risk_exist} is satisfied.
    Also, by Assumption~\ref{assump6_bdd_link}, $\|\mu'\|_{\infty} \leq L$, and so we have that $|\ell'(\zeta,D) - \ell'(\eta,D)| = |\mu(\zeta)-\mu(\eta)|\leq L|\eta-\zeta|$, so that Assumption~\ref{assump4_risk-lipschitz} is met with $\delta = 1$.
    The smoothed score is given by $\phi_P(\eta,x) = \mu(\eta)-\mu(g_P^*(x))$, so that the map $\eta \mapsto \phi_P(\eta,x)$ is differentiable with $\phi'_P(\eta,x) = \mu'(\eta)$.
    Then, $\sup_{\eta \in \mathbb{R}}|\phi'_P(\eta,x)| = \|\mu'\|_{\infty } \leq L$, so that Assumption~\ref{assump4_bdd-weights} is satisfied.
    Also, $w_P(x) = \mu'(f^*_P(x))$, which by Assumption~\ref{assump6_bdd_weights} is at least $k$, so that Assumption~\ref{assump4_lower_bdd-weights} is met.
    The map $\eta \mapsto \phi'_P(\eta,x)$ is equal to $\mu'(\eta)$, and is thus $L'-$Lipschitz as we have $\|\mu''\|_{\infty} \leq L'$, so that Assumption~\ref{assump4_lipschitz-weights} is also met. Finally, we have that $\var_P(\ell'(f^*_P(X),D) \given X = x) = \var_P(Y \given X = x) \in [c,C]$ almost surely by Assumption~\ref{assump6_bdd_variance}, and so combined with Assumption \ref{assump6:bdd_alt}, \cref{assump:score_variance} is satisfied.
\end{proof}

\begin{lemma} \label{lem:quantile-assumptions}
    Consider the conditional quantile specification setup in \cref{sec:quant}, and suppose \cref{assump : quantile} holds.
    Then, the class of distributions $\mathcal{P}$ satisfies \cref{assump:risk_exist,assump:risk-regularity,assump:score_variance} in \cref{sec:theory}.
    \end{lemma}

\begin{proof}
Following the proof of \cref{lem:conditional_quantile}, we have that $\mathbb{E}_P(\ell(g(X),D))$ is well defined and finite whenever $g \in L^2_P(X)$.
The score $\ell'(g(X),D)$ is bounded, and is measurable whenever $g$ is, so that \cref{assump:risk_exist} is satisfied. Now, note that \begin{align*} \mathbb{E}_P(\{\ell'(\eta,D)-\ell'(\zeta,D)\}^2 \given X = x) &= \mathbb{E}_P(|\ind(Y \leq \eta )-\ind(Y \leq \zeta)| \given X = x) \\ &\leq \mathbb{E}_P(\ind(\zeta \leq Y \leq \eta) + \ind(\eta \leq Y \leq \zeta) \given X = x) \\ & = \mathbb{P}_P(\zeta \leq Y \leq \eta \given X = x) + \mathbb{P}_P(\eta \leq Y \leq \zeta \given X = x) \\ &\leq 2L|\eta -\zeta|,
\end{align*}
where in the final step we use the fact that the conditional density is uniformly bounded by $L$.
Thus, Assumption~\ref{assump4_risk-lipschitz} holds with $\delta = 0$.
The smoothed score is given by $\phi_P(\eta,x) = \mathbb{P}_P(Y \leq \eta \given X = x) - \tau $, so that the map $\eta \mapsto \phi_P(\eta,x)$ is differentiable with derivative $\phi'_P(\eta,x) = p_{Y | X}(\eta \given x)$.
Then, $\sup_{\eta \in \mathbb{R}}|\phi'(\eta,x)| = \sup_{\eta \in \mathbb{R}}p_{Y \given X}(\eta \given x) \leq L$, by Assumption~\ref{assump7_bdd_density}, and so Assumption~\ref{assump4_bdd-weights} holds.
Also, $w(x) = p_{Y|X}(f^*_P(x) \given x) \geq k$ by Assumption~\ref{assump7_bdd_density}, so that Assumption~\ref{assump4_lower_bdd-weights} is satisfied.
The map $\eta \mapsto \phi'_P(\eta,x) = p_{Y|X}(\eta|x)$ is $L'$-Lipschitz, and so Assumption~\ref{assump4_lipschitz-weights} is also met.
Note that $\var_P(\ell'(f^*_P(X),D)\given X = x) = \mathbb{P}_P(Y \leq f^*_P(X) \given X = x) \cdot (1-\mathbb{P}_P(Y \leq f^*_P(X) \given X = x)) \in [c(1-C),C(1-c)]$ by Assumption~\ref{assump7_minimiser_exist}, and thus \cref{assump:score_variance} also holds.

\end{proof}

We state the following lemmas proved elsewhere for completeness. We take $\mathcal{P}$ to be a family of distributions that determines distributions of the sequences of random variables in their statements.

\begin{lemma}[{Uniform Slutsky, \citet[Lemma 20]{Shah_Peters_2020}, \citet[Theorem A.11.15]{Samworth_Shah_2026}}] \label{lem:unif_Slutsky}
    Let $(V_n)_{n \in \mathbb{N}}$ and $(W_n)_{n \in \mathbb{N}}$  be sequences of random variables.  Suppose that
	\[
	\lim_{n \to \infty} \sup_{P \in \mathcal{P}} \sup_{t \in \R} \bigl|\pr_P(V_n\leq t) - \Phi(t)\bigr| =0.
	\]
	\begin{enumerate}[(a)]
		\item If $W_n= o_{\mathcal{P}}(1)$, then
        \[
        \lim_{n \to \infty} \sup_{P \in \mathcal{P}} \sup_{t \in \R} \bigl|\pr_P(V_n + W_n \leq t) - \Phi(t)\bigr| =0.
        \]
	\item If $W_n= 1 + o_{\mathcal{P}}(1)$, then
    \[
    \lim_{n \to \infty} \sup_{P \in \mathcal{P}} \sup_{t \in \R} \biggl|\pr_P\Bigl(\frac{V_n}{W_n} \leq t\Bigr) - \Phi(t)\biggr| =0.
		\]
	\end{enumerate}
\end{lemma}

\begin{lemma}[{\citet[Lemma~S6]{lundborg2024projected}}] \label{lem:lems6}
Let $(U_n)_{n \in \mathbb{N}}$ be a sequence of non-negative random variables on $(\Omega,\mathcal{F})$, and let $(\mathcal{F}_n)_{n \in \mathbb{N}}$ be a sequence of sub-$\sigma$-algebras of $\mathcal{F}$. Then if $\E(U_n \given \mathcal{F}_n) = o_{\mathcal{P}}(1)$, then $U_n = o_{\mathcal{P}}(1)$.
\end{lemma}

\begin{lemma}[{Uniform conditional central limit theorem, \citet[Lemma~S8]{lundborg2024projected}}] \label{lem:lems8}
	Let $(X_{n, i})_{n \in \mathbb{N}, i \in [n]}$ be a triangular array of real-valued random variables and let $(\mathcal{F}_n)_{n \in \mathbb{N}}$ be a filtration on $\mathcal{F}$. Assume that
	\begin{enumerate}[(i)]
		\item $X_{n, 1}, \dots, X_{n, n}$ are conditionally independent given $\mathcal{F}_n$, for each $n \in \mathbb{N}$;
		\item $\mathbb{E}_P(X_{n, i} \given \mathcal{F}_n) = 0$ for all $n \in \mathbb{N}, i \in [n]$;
		\item $\bigl| n^{-1} \sum_{i=1}^n \mathbb{E}_P(X_{n, i}^2 \given \mathcal{F}_n) - 1\bigr| = o_\mathcal{P}(1)$;
		\item there exists $\delta > 0$ such that
		\[
			\frac{1}{n}\sum_{i=1}^n \E_{P}\bigl(|X_{n, i}|^{2+\delta} \given \mathcal{F}_n\bigr) = o_{\mathcal{P}}(n^{\delta/2}).
		\]
	\end{enumerate}
	Then $S_n := n^{-1/2} \sum_{m=1}^n X_{n,m}$ converges uniformly in distribution to $\N(0, 1)$, i.e. 
	\[
		\lim_{n \to \infty} \sup_{P \in \mathcal{P}} \sup_{x \in \mathbb{R}} |\mathbb{P}_P(S_n \leq x) - \Phi(x)| = 0.
	\]
\end{lemma}

\subsection{Proofs} \label{subsec:proofs}

\subsubsection{Proof of \cref{thm:asymptotic}} \label{prf:asymptotic}
In the following, we typically suppress dependence on the data distribution $P$ to ease the notational burden.
Throughout, we let $\mathcal{M}_n$ be the filtration generated by the auxiliary datasets $I_2$ and $I_1$ and $D$ denotes an independent observation.
Define the following:
\begin{align*}
\varepsilon &:= \ell'(f^*(X),D), & \xi &:= \hat{h}(X) - m_{\hat{h}}(X),\\
\delta_f &:= \ell'(\hat{f}(X),D) - \ell'(f^*(X),D), &\Delta_m &:= m_{\hat{h}}(X)-\hat{m}_{\hat{h}}(X), \\
\Delta_f &:= \hat{f}(X) - f^*(X), & \Delta_\phi&:= \phi(\hat{f}(X), X) - \phi(f^*(X), X).
\end{align*}
Denote by $\varepsilon_i, \xi_i, \delta_{f,i}, \Delta_{f,i}$ and $\Delta_{m,i}$ the corresponding sample values evaluated on $D_i$, where $i \in I_3$.
Further define $v^2 := \var(\varepsilon\xi\given \mathcal{M}_n)$\footnote{Note that here and throughout, the conditional expectations and variances given $\mathcal{M}_n$ should be regarded as expectations and variances with respect to the regular conditional distribution given $\mathcal{M}_n$, i.e., since $\mathcal{M}_n$ is independent of the data indexed by $I_3$, expectations and variances holding $\mathcal{M}_n$ fixed. For non-negative random variables, this coincides with the usual definition of the conditional expectation. } and note that by \cref{assump:score_variance}, $v^2 \geq \mathbb{E}[\var(\varepsilon \given X)\xi^2 \given \mathcal{M}_n]\geq c\sigma^2$.
Throughout, we shall work on the event $\Omega_0 = \{\sigma^2 > 0\}$, which by \ref{thm1:assump_a} satisfies $\sup_{P \in\mathcal{P}}\mathbb{P}_P(\Omega_0^c) = o(1)$.
Now, denoting $L_i := \ell'(\hat{f}(X_i),D_i)\{h(X_i)-\hat{m}_h(X_i)\}$, let us write
\begin{align*}
T^{(N)} := \frac{1}{\sqrt{nv^2}}\sum_{i=1}^{n}L_i \quad \text{ and } \quad T^{(D)} := \left({\frac{1}{nv^2}\sum \limits_{i=1}^{n} L_i^2 - {\left( \frac{1}{nv}\sum \limits_{i=1}^{n} L_i \right)}^2}\right)^{1/2}.
\end{align*}
Thus the test statistic may be written as $T^{(N)}/T^{(D)}$.

Firstly, we decompose
\begin{align} \label{eq:L_decomp}
L_i = \mathbb{E}(\varepsilon\xi \given\mathcal{M}_n) + \hat{L}_i,
\end{align}
where
\[ \hat{L}_i := \underbrace{\varepsilon_i\xi_i - \mathbb{E}(\varepsilon_i\xi_i \given \mathcal{M}_n)}_{U_i} + \underbrace{\delta_{f,i}\xi_i}_{A_i} + \underbrace{\Delta_{m,i}\varepsilon_i}_{B_i} + \underbrace{\delta_{f,i}\Delta_{m,i}}_{C_i}. \]
Hence, we can write
\begin{align*}
T^{(N)} = \sqrt{n}\cdot \underbrace{\frac{\mathbb{E}(\varepsilon \xi \given \mathcal{M}_n)} {\sqrt{\var(\varepsilon\xi \given \mathcal{M}_n)}}}_{\Lambda } + \hat{T}^{(N)},
\end{align*}
where
\[ \hat{T}^{(N)} := \underbrace{\frac{1}{\sqrt{nv^2}}\sum \limits_{i=1}^{n} U_{i}}_{\Rom{1}_n} + \underbrace{\frac{1}{\sqrt{nv^2}}\sum \limits_{i=1}^{n} A_{i}}_{\Rom{2}_n} + \underbrace{\frac{1}{\sqrt{nv^2}}\sum \limits_{i=1}^{n} B_{i}}_{\Rom{3}_n} + \underbrace{\frac{1}{\sqrt{nv^2}}\sum \limits_{i=1}^{n} C_{i}}_{\Rom{4}_n} . \]
We will show that $\hat{T}^{(N)}$ converges uniformly in distribution on $\mathcal{P}$ to a $N(0,1)$ distribution.
By \cref{lem:unif_Slutsky}, it is enough to show that $\Rom{1}_n$ converges uniformly on $\mathcal{P}$ in distribution to a $N(0,1)$ variable, and that $\Rom{2}_n + \Rom{3}_n + \Rom{4}_n = o_{\mathcal{P}}(1)$.
Note that the $U_i/v$ are standardised and conditionally independent given the filtration $\mathcal{M}_n$.
Moreover, by \ref{thm1:assump_b}, we have that \[ \mathbb{E}(|U_i/v|^{2+\delta} \given \mathcal{M}_n) \leq \frac{2^{\delta +1}}{c^{1+\delta/2}}\cdot\frac{1}{\sigma^{2+\delta}}\mathbb{E}_P({}|\varepsilon \xi|^{2+\delta} \given \mathcal{M}_n) = o_{\mathcal{P}}(n^{\delta/2}),\, \] so that by \cref{lem:lems8} $\Rom{1}_n$ converges uniformly in distribution to a $N(0,1)$ distribution on $\mathcal{P}$.
To see that $\Rom{2}_n = o_{\mathcal{P}}(1)$, notice that by \cref{lem:lems6}, it is enough to show $\mathbb{E}(\Rom{2}^2_n \given \mathcal{M}_n)= o_{\mathcal{P}}(1)$.
We have
\begin{equation} \label{eq:a_n}
\mathbb{E}(\Rom{2}^2_n \given \mathcal{M}_n) = {\mathbb{E}(\Rom{2}_n \given \mathcal{M}_n)}^2 + \text{Var}(\Rom{2}_n \given \mathcal{M}_n).
\end{equation}
Considering the first term in \eqref{eq:a_n}, we have
\begin{equation} \label{eq:II_arg_1}
\begin{aligned} \mathbb{E}(\Rom{2}_n \given \mathcal{M}_n) &= \frac{\sqrt{n}}{v}\mathbb{E}(\delta_f \xi \given \mathcal{M}_n) \\
&= \frac{\sqrt{n}}{v}\mathbb{E}(\underbrace{\mathbb{E}(\delta_f \given X, \mathcal{M}_n)}_{\Delta_{\phi}}\xi \given \mathcal{M}_n) \\
&= \frac{\sqrt{n}}{v} \mathbb{E}(\Delta_{\phi}\xi \given\mathcal{M}_n) \\
&= \frac{\sqrt{n}}{v}\mathbb{E}(w(X)\Delta_f\xi \given \mathcal{M}_n) + \frac{\sqrt{n}}{v}\mathbb{E}(r_n\xi \given\mathcal{M}_n), \end{aligned}
\end{equation}
where $r_n := \Delta_\phi - w(X)\Delta_f$.

By \cref{lem:wls-orth} in \cref{sec:lemmas}, we have the orthogonality condition
\begin{equation*} \mathbb{E}_P\left[\left\{g(X)-m_{g}(X)\right\}w_P(X)f(X)\right] = 0, \quad \forall f \in \mathcal{F}_P,
\end{equation*}
i.e.\ $\mathbb{E}(w(X) \Delta_f\xi \given \mathcal{M}_n) = 0$,
so that $\mathbb{E}(\Rom{2}_n \given \mathcal{M}_n) = \frac{\sqrt{n}}{v} \mathbb{E}(r_n \xi \given \mathcal{M}_n) $.
By \cref{lem:rem_term}, we have that $|r_n| \leq L'\Delta_f^2$, so that
\begin{align*}
\left| \frac{\sqrt{n}}{v}\mathbb{E}(r_n\xi \given \mathcal{M}_n)\right| &\leq \frac{\sqrt{n}}{v}\mathbb{E}(|r_n\xi| \given \mathcal{M}_n)\\
&\leq \frac{\sqrt{n}}{c^{1/2}\sigma}\mathbb{E}(|r_n\xi| \given \mathcal{M}_n)\\
&\leq \frac{L'\sqrt{n} }{c^{1/2}\sigma}\mathbb{E}(|\xi|\Delta_f^2\given \mathcal{M}_n)\\
&\leq \frac{L'}{c^{1/2}}\left\{ n \mathbb{E}(\Delta_f^2\given \mathcal{M}_n)\cdot \mathbb{E}\left(\frac{1}{\sigma^2}\Delta_f^2 \xi^2\given\mathcal{M}_n\right)\right\}^{1/2} \\
&= \frac{L'}{c^{1/2}}\sqrt{n\mathcal{E}_1\mathcal{E}_3} = o_{\mathcal{P}}(1),
\end{align*}
using the Cauchy--Schwarz inequality for the final inequality.
Note that if $\phi_P(\cdot,x)$ is constant for each $x \in \mathcal{X}$, then we can take $L' = 0$ and have that $r_n = 0$.
Thus, the rate condition on $\mathcal{E}_1 \mathcal{E}_3$ can be dropped.

Turning to the second term in \eqref{eq:a_n}, by the conditional independence of $A_i$ and $A_j$ for $i \neq j $ given $\mathcal{M}_n$, the conditional variance term reduces to
\begin{align*}
\var(\Rom{2}_n \given \mathcal{M}_n) &= \frac{1}{v^2}\var(\delta_f \xi \given \mathcal{M}_n)\\
&\leq \frac{1}{v^2}\mathbb{E}(\delta_f^2\xi^2\given\mathcal{M}_n) \\
&\leq \frac{1}{c\sigma^2}\mathbb{E}(\delta_f^2\xi^2\given\mathcal{M}_n) \\
&= \frac{1}{c\sigma^2}\mathbb{E}(\underbrace{\mathbb{E}(\delta_f^2 \given X,\mathcal{M}_n)}_{\leq L|\Delta_f|^{1+\gamma}}\xi^2\given\mathcal{M}_n) \\
&\leq \frac{L}{c\sigma^2}\mathbb{E}(|\Delta_f\xi|^{1+\gamma} |\xi|^{1-\gamma}\given\mathcal{M}_n) \ (\text{By Assumption~\ref{assump4_risk-lipschitz}} ) \\
&\leq \frac{L}{c}\left(\frac{1}{\sigma^2}\mathbb{E}(\Delta_f^2\xi^2\given\mathcal{M}_n) \right)^{\frac{1+\gamma}{2}} \cdot \left(\frac{1}{\sigma^2}\mathbb{E}(\xi^2\given\mathcal{M}_n) \right)^{\frac{1-\gamma}{2}} \ (\text{H\"older's inequality}) \\
&= \frac{L}{c} \mathcal{E}_3^{\frac{1 + \gamma}{2}} = o_{\mathcal{P}}(1),
\end{align*}
so that by \cref{lem:lems6} we have that $\Rom{2}_n = o_{\mathcal{P}}(1)$.

Considering now the term $\Rom{3}_n$, observe that $\mathbb{E}(\Rom{3}_n \given\mathcal{M}_n) = \frac{\sqrt{n}}{v}\mathbb{E}(\Delta_m \varepsilon \given \mathcal{M}_n)$.
As $f^*$ minimises the risk $R_P$ over $\mathcal{F}$ and as $m_h - \hat{m}_h \in \mathcal{F}$, by \cref{lem:risklocmin} we have that $\mathbb{E}(\Delta_m \varepsilon \given \mathcal{M}_n) = 0$.
Thus,
\begin{align*}
\mathbb{E}(\Rom{3}^2_n \given \mathcal{M}_n) &\leq \mathbb{E}\left(\frac{1}{c\sigma^2}\varepsilon^2\Delta_m^2 \given \mathcal{M}_n\right)\\
&= \mathbb{E}\left(\frac{1}{c\sigma^2}\mathbb{E}(\varepsilon^2 \given X)\Delta_m^2 \given \mathcal{M}_n\right) \\
&\leq \frac{C\mathcal{E}_2}{c} = o_{\mathcal{P}}(1).
\end{align*}
Thus by \cref{lem:lems6}, $\Rom{3}_n = o_{\mathcal{P}}(1)$.
Finally, note that
\begin{align*}
\mathbb{E}(|\Rom{4}_n| \given \mathcal{M}_n) = \frac{\sqrt{n}}{v}\mathbb{E}(|\delta_f\Delta_m| \given \mathcal{M}_n) &\leq \frac{\sqrt{n}}{c^{1/2}\sigma}\mathbb{E}(|\delta_f\Delta_m| \given \mathcal{M}_n) \\
& = \frac{\sqrt{n}}{c^{1/2}\sigma}\mathbb{E}(\underbrace{\mathbb{E}(|\delta_f| \given X, \mathcal{M}_n)}_{\leq K|\Delta_f|}|\Delta_m|\given \mathcal{M}_n) \\
&\leq \frac{K\sqrt{n}}{c^{1/2}\sigma} \mathbb{E}(|\Delta_f \Delta_m| \given\mathcal{M}_n) \ (\text{Using \cref{lem:loss_reg}}) \\
& \leq \frac{K}{c^{1/2}} \sqrt{n\mathcal{E}_1 \mathcal{E}_2} = o_{\mathcal{P}}(1)
\end{align*}
where the Cauchy--Schwarz inequality was used in the final step.
It follows from \cref{lem:lems6} that $\Rom{4}_n = o_{\mathcal{P}}(1)$.

We now turn our attention to $T^{(D)}$ and show that $T^{(D)} = 1 + o_{\mathcal{P}}(1)$.
Observe first that since the $L_i$ in the definition of $T^{(D)}$ can be replaced by versions shifted by a constant and leave $T^{(D)}$ unchanged, we have
\begin{equation}
T^{(D)} = \bigg\{ \underbrace{\frac{1}{nv^2}\sum_{i=1}^{n}\hat{L}_{i}^2}_{p_n} - {\bigg(\underbrace{\frac{1}{nv}\sum_{i=1}^{n}\hat{L}_{i}}_{q_n}\bigg)}^{2} \bigg\}^{1/2}.
\end{equation}
Thus noting that $q_n = \frac{1}{\sqrt{n}}\hat{T}^{(N)} = o_{\mathcal{P}}(1)$, it is enough to show that $|p_n - 1| = o_{\mathcal{P}}(1)$.
Now, using \eqref{eq:L_decomp} we can write
\begin{equation}\label{eq:thm1}
\begin{aligned} p_n = &\underbrace{\frac{1}{nv^2}\sum \limits_{i=1}^{n} U_{i}^{2}}_{\tilde{\Rom{1}}_n} + \underbrace{\frac{1}{nv^2}\sum \limits_{i=1}^{n} A_{i}^{2}}_{\tilde{\Rom{2}}_n} + \underbrace{\frac{1}{nv^2}\sum \limits_{i=1}^{n} B_{i}^{2}}_{\tilde{\Rom{3}}_n} + \underbrace{\frac{1}{nv^2}\sum \limits_{i=1}^{n} C_{i}^{2}}_{\tilde{\Rom{4}_n}} \\
&+ \frac{2}{nv^2}\sum \limits_{i=1}^{n} U_{i}A_i + \frac{2}{nv^2}\sum \limits_{i=1}^{n} U_{i}B_i + \frac{2}{nv^2}\sum \limits_{i=1}^{n} U_{i}C_i \\
&+ \frac{2}{nv^2}\sum \limits_{i=1}^{n} A_{i}B_i + \frac{2}{nv^2}\sum \limits_{i=1}^{n} A_{i}C_i \\
&+ \frac{2}{nv^2}\sum \limits_{i=1}^{n} B_{i}C_i . \end{aligned}
\end{equation}

We have that $|\tilde{\Rom{1}}_n - 1| = o_{\mathcal{P}}(1)$ by the uniform weak law of large numbers \citep[Lemma~19]{Shah_Peters_2020}.
It is then enough to show that the other terms in \eqref{eq:thm1} are $o_{\mathcal{P}}(1)$.
Recalling that $|ab| \leq \frac{1}{2}(a^2 + b^2)$, it is in fact enough to show that the remaining three terms in the first row are $o_{\mathcal{P}}(1)$.
To this end, note that
\begin{align*}
\mathbb{E}\left(\frac{1}{nv^2}\sum_{i = 1}^{n}A_i^2 \given \mathcal{M}_n \right) &=\frac{1}{v^2}\mathbb{E}(A_i^2 \given \mathcal{M}_n) \leq\frac{1}{c\sigma^2} \mathbb{E}(\xi^2\delta_f^2 \given\mathcal{M}_n) = o_{\mathcal{P}}(1), \\
\mathbb{E}\left(\frac{1}{nv^2}\sum_{i = 1}^{n}B_i^2 \given \mathcal{M}_n \right) &=\frac{1}{v^2}\mathbb{E}(B_i^2 \given \mathcal{M}_n) \leq \frac{1}{c\sigma^2} \mathbb{E}(\varepsilon^2\Delta_m^2 \given\mathcal{M}_n) = o_{\mathcal{P}}(1).
\end{align*}
Thus by \cref{lem:lems6}, the terms $\tilde{\Rom{2}}_n$ and $\tilde{\Rom{3}}_n$ are $o_{\mathcal{P}}(1)$.
Also,
\begin{align*}
\frac{1}{nv^2}\sum_{i = 1}^{n}C_i^2 \leq \left(\frac{1}{\sqrt{n}v}\sum_{i = 1}^{n}|C_i|\right)^2.
\end{align*}
The term on the right in the above display is $o_{\mathcal{P}}(1)$ as from our earlier analysis of the term $\Rom{4}_n$ we have
\begin{align*}
\mathbb{E}\left(\frac{1}{\sqrt{n}v}\sum_{i = 1}^{n}|C_i|\given \mathcal{M}_n\right ) \leq\frac{\sqrt{n}}{c^{1/2}\sigma}\mathbb{E}(|\delta_f\Delta_m|\given \mathcal{M}_n) = o_{\mathcal{P}}(1).
\end{align*}
It follows that $\tilde{\Rom{4}}_n $ is also $o_{\mathcal{P}}(1)$, and hence that $T^{(D)} = 1 + o_{\mathcal{P}}(1)$.
We can therefore conclude that the test statistic
\[ T_n = \frac{\sqrt{n}\Lambda + \hat{T}^{(N)}}{T^{(D)}} = \sqrt{n}\Lambda \cdot (1+r_n) +Z_n, \]
where $Z_n := \hat{T}^{(N)}/T^{(D)}$ converges uniformly in distribution to $N(0,1)$ on $\mathcal{P}$ by \cref{lem:unif_Slutsky} and $r_n:= \frac{1}{T^{(D)}}-1 = o_{\mathcal{P}}(1).$ \qed

\subsubsection{Proof of \cref{cor:power}} \label{prf:power}
\begin{proof}
By \cref{thm:asymptotic}, for every $P\in\mathcal{P}(\epsilon_n)$,
\[
T_n=Z_n+\sqrt{n}\,(1+r_n)\Lambda_{n,P},
\]
where $r_n=o_{\mathcal{P}(\epsilon_n)}(1)$ and
\[
\sup_{P\in\mathcal{P}(\epsilon_n)}\sup_{t\in\mathbb{R}}\left|\mathbb{P}_P(Z_n\leq t)-\Phi(t)\right|\to0.
\]

Define
\[
E_{n,P}:=\left\{\operatorname{Corr}_P(s_P(X),\xi_P\given\hat{h})>\rho_n\right\},
\]
and note that since $P\in\mathcal{P}(\epsilon_n)$ implies $\tau_P\geq\epsilon_n$, on $E_{n,P}$, \cref{lem:snr_bound} yields
\[
\sqrt{n}\,\Lambda_{n,P}\geq A\sqrt{n\epsilon_n}\,\rho_n.
\]
Set $a_n:=A\sqrt{n\epsilon_n}\,\rho_n$, and note that since $n\rho_n^2\epsilon_n\to\infty$, we have $a_n\to\infty$. Further, since $r_n=o_{\mathcal{P}}(1)$, $
\sup_{P\in\mathcal{P}(\epsilon_n)}\mathbb{P}_P(r_n<-1/2)\to 0.$
Hence, on the event $G_{n,P}:=E_{n,P}\cap\{r_n\geq-1/2\}$,
we have that 
\[
\sqrt{n}\,(1+r_n)\Lambda_{n,P}\geq\frac{a_n}{2}.
\]
By \eqref{eq:quality_of_hunting}, $
\inf_{P\in\mathcal{P}(\epsilon_n)}\mathbb{P}_P(E_{n,P})\to1,
$
and therefore $\inf_{P\in\mathcal{P}(\epsilon_n)}\mathbb{P}_P(G_{n,P})\to1.$ For any $P\in\mathcal{P}(\epsilon_n)$,
\begin{align*}
\mathbb{P}_P(T_n\leq z_{1-\alpha})
&\leq\mathbb{P}_P\left(Z_n+\frac{a_n}{2}\leq z_{1-\alpha}\right)+\mathbb{P}_P(G_{n,P}^c)\\
&=\mathbb{P}_P\left(Z_n\leq z_{1-\alpha}-\frac{a_n}{2}\right)+\mathbb{P}_P(G_{n,P}^c).
\end{align*}
Hence,
\begin{align*}
\sup_{P\in\mathcal{P}(\epsilon_n)}\mathbb{P}_P(T_n\leq z_{1-\alpha})
&\leq\Phi\left(z_{1-\alpha}-\frac{a_n}{2}\right)\\
&\quad+\sup_{P\in\mathcal{P}(\epsilon_n)}\sup_{t\in\mathbb{R}}\left|\mathbb{P}_P(Z_n\leq t)-\Phi(t)\right|\\
&\quad+\sup_{P\in\mathcal{P}(\epsilon_n)}\mathbb{P}_P(G_{n,P}^c).
\end{align*}
The first term converges to zero since $a_n\to\infty$, the second by \cref{thm:asymptotic}, and the third by the preceding argument. Therefore,
\[
\sup_{P\in\mathcal{P}(\epsilon_n)}\mathbb{P}_P(T_n\leq z_{1-\alpha})\to0,
\]
or equivalently,
\[
\inf_{P\in\mathcal{P}(\epsilon_n)}\mathbb{P}_P(T_n>z_{1-\alpha})\to1.
\]
\end{proof}
\section{Implementation details and additional numerical results} \label{sec:simu-extra}
\subsection{Discussion of \texttt{anova.gam} from the \texttt{mgcv} package \citep{wood2017}} \label{subsec:mgcv}
The \texttt{anova.gam} function from the \texttt{mgcv} package in \texttt{R} provides a variety of hypothesis tests pertaining to fitted additive models.
In particular, it can be used to compare nested models using a generalised likelihood ratio test (GLRT).
Suppose $\mathcal{M}_0 \subset \mathcal{M}_1$ are the nested models under consideration, and $\hat{\beta}_0$ and $\hat{\beta}_1$ are the estimates of the basis coefficients under each model.
Then, the test statistic is calculated as
\[ \Lambda = 2\{\ell(\hat{\beta}_1) - \ell(\hat{\beta}_0)\}, \]
where $\ell$ is the log-likelihood function used when fitting.
Motivated by the classical GLRT for generalised linear models \citep{glm}, the test statistic $\Lambda$ is compared to a $\chi^2_{\text{EDF}_1-\text{EDF}_0}$ distribution, where $\text{EDF}_j$ is the effective degrees of freedom of the model $\mathcal{M}_j$.
In principle, one could use the above to assess the goodness of fit of an additive model by comparing the additive model to a bigger model with interaction terms.
For example, when $p = 2$, it is computationally feasible to compare the additive model $y \sim s(x_1) +s(x_2)$ to the unrestricted model by incorporating the interaction term $\texttt{ti}(x_1,x_2)$ into the purely additive model.

The \texttt{mgcv} documentation notes that the $p$-values resulting from this procedure are only approximate, and are usually too low in simulations.
Moreover, it is possible to have settings where the effective degrees of freedom of the larger model $\mathcal{M}_1$ is significantly \textit{smaller} than that of the smaller model $\mathcal{M}_0$.
The chi-squared approximation above then becomes nonsensical, but if the larger model $\mathcal{M}_1$ also happens to fit the data better (i.e. has a larger likelihood), then it is reasonable to reject the null model $\mathcal{M}_0$.
However, it is possible to have settings where the null model $\mathcal{M}_0$ is well-specified, yet the fitted $\mathcal{M}_1$ has fewer effective degrees of freedom \textit{and} fits the data better, leading to an incorrect rejection of the null model.
Note that the DST is still calibrated in this setting.

For instance, consider the following data-generating mechanism; the covariates are given by $X_1 \sim \operatorname{Uniform}[1,2]$, $X_2 = \frac{1}{X_1}(f_a(X_1) + 0.1 \cdot \eta)$ with the response $Y = f_a(X_1) + \varepsilon$, where $\eta$ and $\varepsilon$ are independent $N(0,1)$ variables, and $f_a(x) = \sin(2\pi a(x-1))$.
As $a \in \{1,2,3\}$ increases, $f_a$ becomes increasingly more wiggly.
Since $X_1 X_2 \approx f_a(X_1)$, it can then take fewer degrees of freedom to explain more variation in $Y$ by using the larger model $\mathcal{M}_1$.
Consequently, one expects \texttt{anova.gam} to fail to control Type I error.
This is precisely what we observe empirically; see \cref{tab:anova.gam_calibration}.

\begin{table}[htbp]
\centering
\caption{Empirical rejection rates (level $\alpha=0.05$) under the null additive model for increasing smoothness parameter $a$.}
\label{tab:anova.gam_calibration}
\begin{tabular}{ccc}
\toprule
$a$ & GLRT (anova.gam) & Debiased Score Test (DST) \\
\midrule
1 & 0.106 & 0.066 \\
2 & 0.322 & 0.042 \\
3 & 0.540 & 0.042 \\
\bottomrule
\end{tabular}
\end{table}

\subsection{Adaptation of the test from \citet{williamson2021general}} \label{sec:williamson}
\citet{williamson2021general} formulate variable importance in terms of a model-agnostic, population-level parameter by defining the population-level importance of a variable $X_s$ relative to the full covariate vector $X$ to be
\[ \psi_{0,s}:= V(f_0,P_0)-V(f_{0,s},P_0), \]
where $f \mapsto V(f,P)$ is a measure of the fit of a candidate prediction function $f$, $f_0 = \argmax_{f \in \mathcal{F}}V(f,P_0)$ and $f_s = \argmax_{f \in \mathcal{F}_s}V(f,P_0)$, where $\mathcal{F}_s$ is those functions in $\mathcal{F}$ which do not depend on $X_s$, while $\mathcal{F}$ is a rich class of functions of $X$.
They then show that under suitable conditions, the value measure $v_0:= V(f_0,P_0)$ can be estimated using the plug-in estimate $v_n = V(f_n,P_n)$, where $f_n$ is an estimate of $f_0$, and that $v_n$ is non-parametric efficient.
Similarly, one can estimate $v_{n,s}$ and hence $\psi_{0,s}$ by $\psi_{n,s} = v_n-v_{n,s}$, where the variance of $\psi_{n,s}$ is estimated using $\tau^2_{n,s} = \frac{1}{n}\sum_{i = 1}^{n}\dot{V}(f_n,P_n;\delta_{Z_i}-P_n)-\dot{V}(f_{n,s},P_0;\delta_{Z_i}-P_n)$.
Then, when $\psi_{0,s} > 0$, they show that $T_n = \frac{\sqrt{n}(\psi_{0,s}-\psi_{n,s})}{\tau_{n,s}} \to N(0,1)$, which in particular allows one to do inference for $\psi_{0,s}$.
In practice, the authors recommend a cross-fitted version of the test statistics, where the estimates $(f_n,f_{n,s})$ and the test statistic $T_n$ are formed on complementary folds.

However note that the influence function of $\psi_{0,s}$ is given by $\varphi_{0,s} : z \mapsto \dot{V}(f,P_0;\delta_z-P_0)-\dot{V}(f_{s},P_0;\delta_z-P_0)$, which when $\psi_{0,s} = 0$ is identically zero.
In this case, $\psi_{n,s}$ generally does not tend to a non-degenerate law, and in particular we cannot test the zero-importance null hypothesis using the test statistic above.
The authors recommend sample splitting to mitigate this, so that one calculates $v_{n,0}$ and $v_{n,s}$ on independent splits, and then takes the difference.
The cross-fitted version is analogous, and proceeds by using independent folds for estimating $v_{0}$ and $v_{0,s}$.

We modify this procedure to adapt to the goodness-of-fit setting.
We define $V(f,P):= -R_P(f)$, where the risk function $R_P$ is as defined in \cref{sec:dst}, and the population-level goodness-of-fit statistic for a class of functions $\mathcal{G}$ to be $v_{\mathcal{G}} = V(g^*_P,P)$, where $g^*_P$ is the minimiser of the risk $R_P$ over $\mathcal{G}$.
Then, the goodness of fit of $\mathcal{F}$ relative to $\mathcal{G}$ may be measured by
\[ \psi = v_{\mathcal{G}}-v_{\mathcal{F}}. \]
Again, one can estimate $v_{\mathcal{F}}$ and $v_\mathcal{G}$ using plug-in estimates, exactly as above.
We are interested in testing the hypothesis $\psi = 0$, and so we estimate them on independent splits of the sample.
In particular, we use \cref{alg:adapted_williamson_test} to carry out the test.

\begin{algorithm}[htbp]
\caption{Adaptation of the test from \citet{williamson2021general}.}
\label{alg:adapted_williamson_test}
\begin{algorithmic}[1]
\State Generate $B_n \in \{1,\ldots,2K\}^n$ by sampling uniformly from $\{1,2,\ldots,2K\}$ with replacement, and for $j = 1,\ldots,2K$, denote by $D_j$ the set of observations with index in $S_j := \{i : B_{n,i} = j\}$ and $n_j := |D_j|$;
\For{$k = 1,\ldots,2K$}
    \State using only data in $\cup_{j \neq k} D_j$, construct estimators $f_{k,n}$ of $f_0$ and $g_{k,n}$ of $g_{0}$;
    \State using only data in $D_k$, construct estimator $P_{n,k}$ of $P_0$;
    \State if $k$ is odd, compute $\eta^2_{\mathcal{G},k,n} := \frac{1}{n_k}\sum_{i \in S_k} \dot{V}(g_{k,n}, P_{k,n};\, \delta_{Z_i} - P_{k,n})^2$ and $v_{\mathcal{G},k,n} := V(g_{k,n}, P_{k,n})$;
    \State if $k$ is even, compute $\eta^{2}_{\mathcal{F},k,n} := \frac{1}{n_k}\sum_{i \in S_k} \dot{V}(f_{k,n}, P_{k,n};\, \delta_{Z_i} - P_{k,n})^2$ and $v_{\mathcal{F}, k,n} := V(f_{k,n}, P_{k,n})$;
\EndFor
\State Compute $v^*_{\mathcal{G},n} := \frac{1}{K}\sum_{k=1}^{K} v_{\mathcal{G},2k-1,n}$, $v^*_{\mathcal{F},n} := \frac{1}{K}\sum_{k=1}^{K} v_{\mathcal{F},2k,n}$ and estimator $\psi^*_{n} := v^*_{\G,n} - v^*_{\F,n}$ of $\psi$;
\State Compute $\eta^2_{\G,n} := \frac{1}{K}\sum_{k=1}^{K} \eta^2_{\G,2k-1,n}$, $\eta^2_{\F,n} := \frac{1}{K}\sum_{k=1}^{K} \eta^2_{\F,2k,n}$ and estimator $\omega_{n} := \eta^2_{\G,n}/(n-n_{\F}) + \eta^2_{\F,n}/n_\F$ of the variance of $\psi^*_n$.
\State Compute $p$-value $p_n =1-\Phi(T_n)$, where $T_n := \omega^{-1/2}_{n}\psi^*_{n}.$
\end{algorithmic}
\end{algorithm}
\subsection{GRF calibration test} \label{subsec:grf}
The \texttt{grf} package \citep{grf} natively supports a way of testing the heterogeneity of the treatment effect.
The method \texttt{test\_calibration()} tests the null hypothesis which posits that the coefficient
\begin{equation} \label{eqn:cond_cov}
\beta(Z) := \frac{\operatorname{Cov}(Y,T \given Z)}{\var(T \given Z)}
\end{equation}
is constant in $Z$.
Note that one can write
\[ Y -\mathbb{E}(Y \given Z) = \underbrace{\{\mathbb{E}(\beta(Z))\}\cdot \{T-\mathbb{E}(T\given Z)\}}_{\text{mean prediction}}+\underbrace{\{\beta(Z)-\mathbb{E}(\beta(Z))\} \cdot \{T-\mathbb{E}(T\given Z)\}}_{\text{differential prediction}} + \varepsilon, \]
where $\mathbb{E}(\varepsilon \given Z) = 0$.
Thus given (out-of-bag) estimates $\hat{\beta}$, $\hat{Y}$ and $\hat{T}$ of $\beta$, $\mathbb{E}(Y \given Z)$ and $\mathbb{E}(T \given Z)$ respectively, one can attempt to test the null above by fitting the linear model
\[ Y_i - \hat{Y}_i \sim \underbrace{\bar{\beta} \cdot (T_i - \hat{T}_i)}_{\text{mean prediction}} + \underbrace{\hat{\beta}(Z_i) \cdot (T_i - \hat{T}_i)}_{\text{differential prediction}}, \]
where $\bar{\beta}$ is the sample mean of $\hat{\beta}$, and then testing the hypothesis whether the coefficient of the differential prediction term is zero.
This is precisely the approach of the \texttt{test\_calibration} function, where all estimates are obtained by fitting a \texttt{causal\_forest} to the data.

Note that the null in \eqref{eqn:cond_cov} is equivalent to the one considered in \cref{sec:hte} when the treatment $T$ is binary.
However, this is not the case when the treatment is continuous.
For instance, suppose that $Y = T^3 + \varepsilon_1$, with $T = \varepsilon_2 \cdot Z$, where $Z, \varepsilon_1$ and $\varepsilon_2$ are all independent and $N(0,1)$.
Then, there is no treatment effect heterogeneity in the sense that $\mathbb{E}(Y \given Z, T) = T^3$ is independent of the control variable $Z$.
Yet, it may be checked that $\beta(Z) = 3Z^2$.

Similarly, there are alternatives to the null we consider for which $\beta(Z)$ is constant.
For instance, consider $Y = T^2Z +\varepsilon$, where $\varepsilon,Z$ and $T$ are all independent $N(0,1)$ variables.
Then, $\mathbb{E}(Y \given T = t,Z) - \mathbb{E}(Y \given T = 0,Z) = t^2Z$ depends on $Z$, yet $\beta(Z) = 0$.

This can also be seen empirically.
We generate $n = 2000$ samples from the null model above, and calculate the rejections rates of \texttt{test\_calibration} and the DST.
For the DST, we use \texttt{mgcv} to fit the null model with $k = 20$.
The rejection rates are presented in \cref{tab:cont-t}.

\begin{table}[htbp] 
\centering
\caption{Empirical size at the 5\% nominal level under the null hypothesis.}
\label{tab:cont-t}
\begin{tabular}{lcc}
\toprule
 & DST & grf \\
\midrule
Rejection Rate & 0.056 & 1.000 \\
\bottomrule
\end{tabular}
\end{table}

Likewise, we consider $n = 2000$ samples from a series of alternative models, with $Y = \frac{1}{4}\{(1-\alpha) T^2Z + \alpha\cdot TZ\} + \varepsilon$, where $\varepsilon,Z$ and $T$ are all independent $N(0,1)$ variables.
Note that here $\beta(Z) = \frac{\alpha Z}{4}$, and we consider $\alpha \in \{0,0.2,0.4,0.6,0.8,1\}$.
The results are noted in \cref{tab:more-alpha}.

\begin{table}[htbp]
\centering
\caption{Rejection rates for DST and GRF across values of $\alpha$.}
\label{tab:more-alpha}
\begin{tabular}{l cccccc}
\hline
$\alpha$ & 0.0 & 0.2 & 0.4 & 0.6 & 0.8 & 1.0 \\
\hline
DST & 1.000 & 0.998 & 0.996 & 0.982 & 0.996 & 1.000 \\
GRF & 0.110 & 0.264 & 0.692 & 0.988 & 1.000 & 1.000 \\
\hline
\end{tabular}
\end{table}

\subsection{Additional results on the generalised additive model} \label{sec:gam-simu-extra}
\cref{fig:gam-power-full} presents the power curves under $n \in \{1000,2000,5000\}$ for the simulation in \cref{sec:gam-simu}.

\begin{figure}[htbp]
    \centering
    \includegraphics[width= 0.9\textwidth]{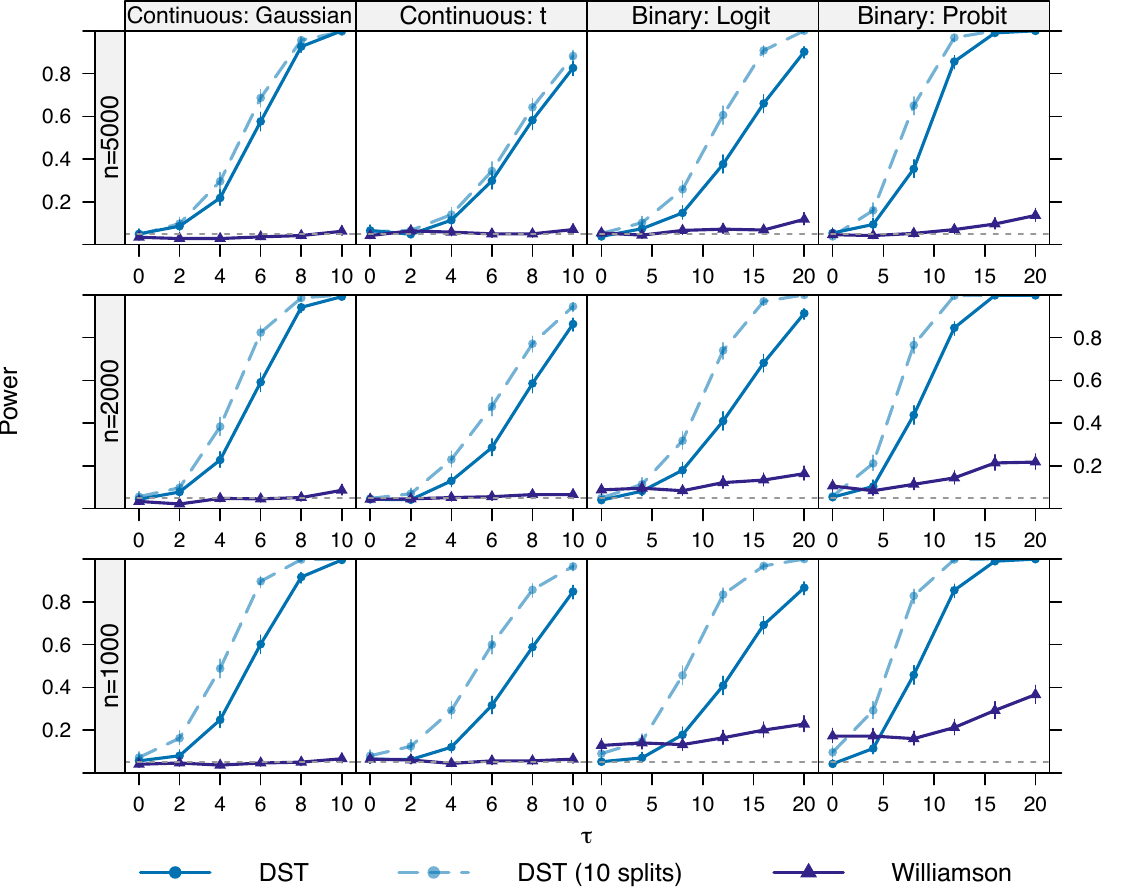}
    \caption{Power (with 95\% CI, dashed horizontal: $\alpha=0.05$) for testing the conditional mean specification of a GAM with $p=10$ covariates. Dashed curves are from combining 10 splits of DST using the rank-transformed subsampling. }
    \label{fig:gam-power-full}
\end{figure}

\subsection{Comparison with unmodified TE-VIM test} \label{sec:tevim}
In the same setting as \cref{sec:hte-numerical}, \cref{fig:hte-tevim-unmod} plots the power curves, where the TE-VIM here is unmodified and thus has no Type I error guarantee under the null. 
Consequently, the rank-transformed subsampling, which requires a known asymptotic null distribution, is not applicable to this version of TE-VIM.

\begin{figure}[htbp]
\centering
\includegraphics[width=.7\textwidth]{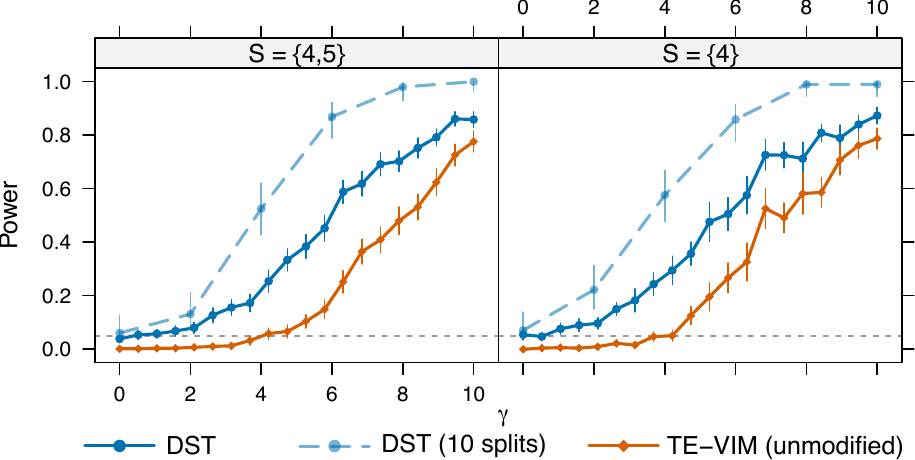}
\caption{Power (with 95\% CI) for detecting treatment effect heterogeneity in covariates $S$ (horizontal dashed: $\alpha = 0.05$). The TE-VIM here is unmodified, which in this particular case is conservative under the null. Dashed curves are from combining 10 splits of DST using the rank-transformed subsampling.}
\label{fig:hte-tevim-unmod}
\end{figure}

\subsection{Generalised partially linear model} \label{sec:plm-simu}
We use our method in \cref{sec:cond-mean} to test the specification of a generalised partially linear model. 
We draw covariates $X \sim N_{20}(0,\Sigma)$ with $\Sigma_{ij} = 2^{-|i-j|}$. 
Conditioned on $X$, we generate binary response $Y$ with the probit link:
\begin{align*}
& \mathbb{P}(Y=1\given X) = \Phi\left(\beta^TU + f(V) + (\tau / \sqrt{n}) \cdot g(U,V) \right), \quad f(V) = \sin(V_{2})\cos(V_1), \\
& \beta = (1/2,-1/2, \ldots, 1/2,-1/2) / 3\sqrt{2},
\end{align*}
where $U= (X_1,X_2, \ldots,X_{18})$ and $V = (X_{19},X_{20})$.
The effect size $\tau \geq 0$ controls the departure from the null, so a partially linear model with linear component $U$ and nonlinear component $V$ is well-specified if and only if $\tau = 0$. 
We consider two forms of the alternative: (i) $g(U,V) = U_1V_1$ and (ii) $g(U,V) = \cos(2U_1)$.
The null model is fit using the package \texttt{mgcv}.
\cref{fig:plm-power} plots the power curves under $n=2000$. 

\begin{figure}[htbp]
    \centering
    \includegraphics[width=0.75\textwidth]{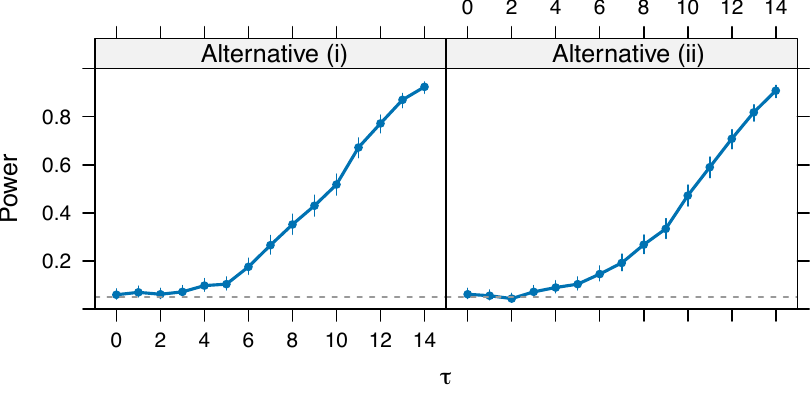}
    \caption{Power (with 95\% CI) of DST for testing the specification of a partially linear model under sample size $n=2000$ (dashed horizontal line: $\alpha=0.05$).}
    \label{fig:plm-power}
\end{figure}

\end{document}